\documentclass[10pt, journal]{IEEEtran}
\usepackage{cite}
\usepackage{amsmath,amssymb}    
\usepackage{graphicx}   
\usepackage{verbatim}   
\usepackage{color}      
\usepackage{subfigure}  
\newcommand{\urvh}{{\rvh}}

\newcommand{\burvg}{{{\rvbg}}}
\newcommand{\brvbz}{\bar{\rvbz}}
\newcommand{\brvby}{\bar{\rvby}}
\newcommand{\brvbv}{\bar{\rvbv}}

\newcommand{\brvbx}{\bar{\rvbx}}
\newcommand{\brvbn}{\bar{\rvbn}}

\newcommand{\hh}{\hat{\rvh}}

\newcommand{\cA}{{\mathcal{A}}}

\newcommand{\cC}{{\mathcal{C}}}

\newcommand{\cE}{{\mathcal{E}}}

\newcommand{\cJ}{{\mathcal{J}}}

\newcommand{\cK}{{\mathcal{K}}}

\newcommand{\CN}{{\mathcal{CN}}}

\newcommand{\cP}{{\mathcal{P}}}

\newcommand{\cT}{{\mathcal{T}}}

\newcommand{\al}{\alpha}

\newcommand{\g}{\gamma}

\newcommand{\eps}{\varepsilon}

\usepackage{verbatim} 
 \usepackage{amsxtra}  
\newcounter{actr}
{\begin{list}{(\alph{actr})}{\usecounter{actr}}}{\end{list}}

\newcounter{ictr}
{\begin{list}{(\roman{ictr})}{\usecounter{ictr}}}{\end{list}}

\newtheorem{remark}{Remark}
\newtheorem{thm}{Theorem}
\newtheorem{lemma}{Lemma}

\newtheorem{corol}{Corollary}
\newtheorem{prop}{Proposition}
\newtheorem{defn}{Definition}

\newenvironment{new-proof}[1]
{{\em Proof  #1: }}%
{ \noindent\qed }
\newcommand{\qed}{\rule[0.1ex]{1.4ex}{1.6ex}}

\newcommand{\defeq}{\stackrel{\Delta}{=}}

\newcommand{\mrm}{\mathrm}

\newcommand{\floor}[1]{\lfloor{#1}\rfloor}

\newcommand{\T}{{\mathrm{T}}}

\DeclareMathAlphabet{\mathbsf}{OT1}{cmss}{bx}{n}
\DeclareMathAlphabet{\mathssf}{OT1}{cmss}{m}{sl}

\DeclareSymbolFont{bsfletters}{OT1}{cmss}{bx}{n}  
\DeclareSymbolFont{ssfletters}{OT1}{cmss}{m}{n}
\DeclareMathSymbol{\bsfGamma}{0}{bsfletters}{'000}
\DeclareMathSymbol{\ssfGamma}{0}{ssfletters}{'000}
\DeclareMathSymbol{\bsfDelta}{0}{bsfletters}{'001}
\DeclareMathSymbol{\ssfDelta}{0}{ssfletters}{'001}
\DeclareMathSymbol{\bsfTheta}{0}{bsfletters}{'002}
\DeclareMathSymbol{\ssfTheta}{0}{ssfletters}{'002}
\DeclareMathSymbol{\bsfLambda}{0}{bsfletters}{'003}
\DeclareMathSymbol{\ssfLambda}{0}{ssfletters}{'003}
\DeclareMathSymbol{\bsfXi}{0}{bsfletters}{'004}
\DeclareMathSymbol{\ssfXi}{0}{ssfletters}{'004}
\DeclareMathSymbol{\bsfPi}{0}{bsfletters}{'005}
\DeclareMathSymbol{\ssfPi}{0}{ssfletters}{'005}
\DeclareMathSymbol{\bsfSigma}{0}{bsfletters}{'006}
\DeclareMathSymbol{\ssfSigma}{0}{ssfletters}{'006}
\DeclareMathSymbol{\bsfUpsilon}{0}{bsfletters}{'007}
\DeclareMathSymbol{\ssfUpsilon}{0}{ssfletters}{'007}
\DeclareMathSymbol{\bsfPhi}{0}{bsfletters}{'010}
\DeclareMathSymbol{\ssfPhi}{0}{ssfletters}{'010}
\DeclareMathSymbol{\bsfPsi}{0}{bsfletters}{'011}
\DeclareMathSymbol{\ssfPsi}{0}{ssfletters}{'011}
\DeclareMathSymbol{\bsfOmega}{0}{bsfletters}{'012}
\DeclareMathSymbol{\ssfOmega}{0}{ssfletters}{'012}

\renewcommand{\defeq}{\triangleq}

\newcommand{\rve}{{\mathssf{e}}}	

\newcommand{\rvg}{{\mathssf{g}}}	

\newcommand{\rvbg}{{\mathbsf{g}}}

\newcommand{\rvh}{{\mathssf{h}}}	

\newcommand{\rvk}{{\mathssf{k}}}	

\newcommand{\rvm}{{\mathssf{m}}}	

\newcommand{\rvn}{{\mathssf{n}}}	
\newcommand{\rvbn}{{\mathbsf{n}}}

\newcommand{\rvq}{{\mathssf{q}}}	

\newcommand{\rvs}{{\mathssf{s}}}	

\newcommand{\rvu}{{\mathssf{u}}}	

\newcommand{\rvbu}{{\mathbsf{u}}}

\newcommand{\rvv}{{\mathssf{v}}}	

\newcommand{\rvbv}{{\mathbsf{v}}}

\newcommand{\rvw}{{\mathssf{w}}}	

\newcommand{\rvx}{{\mathssf{x}}}	

\newcommand{\rvbx}{{\mathbsf{x}}}

\newcommand{\rvy}{{\mathssf{y}}}	

\newcommand{\rvby}{{\mathbsf{y}}}

\newcommand{\rvz}{{\mathssf{z}}}	%

\newcommand{\rvbz}{{\mathbsf{z}}}

\title{ \vspace{0.7em}   Secret-Key Agreement Capacity over Reciprocal Fading Channels: A Separation Approach}

{\author{\IEEEauthorblockN{Ashish Khisti}\\
\IEEEauthorblockA{Dept. of Electrical and Computer Engineering\\
University of Toronto\\
Toronto, ON, Canada\\
Email: akhisti@comm.utoronto.ca}\thanks{This work was supported by a Discovery Grant from the Natural Sciences Engineering Research
Council (NSERC) of Canada.}}}

\begin{document}

\maketitle
\begin{abstract}
Fundamental limits of secret-key agreement over reciprocal wireless channels are investigated.
We consider a two-way block-fading channel where the channel gains in the forward and reverse links 
between the legitimate terminals are correlated.
The channel gains  between the legitimate terminals are not revealed to any terminal, 
whereas the channel gains of the eavesdropper are revealed  to the eavesdropper. 
We propose a two-phase transmission scheme, that reserves a certain portion of each coherence block for channel estimation,
and  the remainder of the coherence block for  correlated source generation.
The resulting secret-key  involves  contributions of both channel sequences  and source sequences, with the contribution of the latter
becoming dominant as the coherence period increases. We  also establish an upper bound on the secret-key capacity, which has a form structurally similar to the lower bound. Our upper and lower bounds coincide in the limit of high signal-to-noise-ratio (SNR) and large coherence period, thus establishing the secret-key agreement capacity in this asymptotic regime. Numerical results indicate that the proposed scheme  achieves significant gains over training-only schemes, even for moderate SNR and small coherence periods, thus implying the necessity of randomness-sharing in practical secret-key generation systems.
\end{abstract}

\section{Introduction}

Physical layer can provide a valuable resource for enhancing security and confidentiality of wireless communication systems.
One promising method that has received a significant interest in recent years (see e.g.,~\cite{wilsonTseScholtz:07,mathur-10,ye-10,croft,ko-uwb,pap1, jana,pap3, pap4,ZaferAS12,VasudevanGT10,mohapatra,mohapatra2,ZhangKP10,MadisehNM12,wangSu,Hamida,Shimizu,Madiseh-thesis,WallaceS10}
and the references therein) is the generation of a shared secret key  using channel reciprocity. 
When the terminals  use identical carrier frequencies, a physical reciprocity~\cite{tseViswanath:05} 
exists between uplink and downlink channels, which provides a natural mechanism for key exchange. The terminals first transmit training signals to estimate the underlying channel gains, then agree on a common sequence using error correction, and finally distill a shared secret-key from this 
common sequence. A number of experimental testbeds already demonstrate practical viability of such methods. See e..g.,~\cite{mohapatra2,Shimizu,Madiseh-thesis,WallaceS10}. 

Despite this growing interest, information theoretic limits of secret-key agreement over wireless channels have not received much attention. Consider a  block-fading channel where the channel gains are sampled once every coherence period and remain constant during the coherence block.  
Clearly, channel reciprocity alone achieves a rate that vanishes as the coherence period increases. What form of signalling within each coherence block  maximizes the secret-key rate? To our knowledge, only a recent paper~\cite{Lai2012} considers this question. The authors consider a separation based scheme. A portion of each coherence block is reserved for transmission of pilot symbols to estimate the channel coefficients. The remainder of the coherence block is used for transmission of a {\em confidential-message}.  Unfortunately the authors observe that one cannot separately tune the two phases. Any rate/power adaptation done during the message transmission phase, leaks information about the channel gains to an eavesdropper and in turn reduce the contribution from reciprocity. As such the wiretap code for confidential message transmission proposed in~\cite{Lai2012} 
does not involve any rate or power adaptation. It achieves a non-zero rate only if we impose that the eavesdropper's channel  on average be weaker than the  legitimate receiver's channel. If this condition is not satisfied, the secret-key rate achieved in~\cite{Lai2012} reduces to the rate achieved using a training only scheme and vanishes to zero as the coherence period becomes large. 

In the present paper we too consider a separation based scheme, but  substitute the confidential-message transmission phase with a different communication strategy  --- {\em randomness sharing}. Following the training phase, in each coherence block the terminals exchange i.i.d.\  random variables to generate correlated source sequences between the terminals. No power allocation is performed during this phase, yet it results in a positive secret-key rate even when the eavesdropper's channel is on average stronger than the legitimate receiver's channel.   We note that the {\em randomness sharing} scheme for secret-key generation has been investigated in several previous works (see e.g.,~\cite{maurer:93,ahlswedeCsiszar:93,CsiszarNarayan:04,Gohari:1,
Gohari:2,prabhakaran:09} and references therein).  However to the author's knowledge, the present  paper appears to be the first attempt that studies randomness sharing, in conjunction with training techniques, in non-coherent wireless channels. Furthermore our proposed scheme has several appealing features. First it achieves a rate that does not vanish in the large coherence-period limit, thus making it attractive in this practically relevant regime. Secondly our proposed separation scheme is asymptotically optimal  among all signalling schemes in the two-way fading channel.  In particular, we also establish an upper bound on the secret-key agreement capacity in the proposed setup.  The  upper bound expression is structurally similar to the lower bound expression and furthermore the difference between our upper and lower bound vanishes to zero in the limit of large coherence period and high signal-to-noise-ratio, thus establishing an asymptotic capacity result. Numerical results indicate that even for moderate SNR and relatively small coherence periods the gains resulting from randomness sharing are significant over training-only schemes. 

Our proposed model also differs from reference~\cite{Lai2012} in that we only assume {\em{imperfect}} reciprocity in uplink and downlink channels i.e., we assume that the channel gains in uplink and downlink are correlated, but not identical. In general achieving perfect channel reciprocity in baseband is  challenging because  different terminals use different I/Q mixers, amplifiers and path lengths in the RF chains.
While  closed-loop calibration can be performed (see e.g.~\cite{haile,tdd-receip-1}), such methods can become challenging if the calibration needs to be done in the open air. Hence we believe that our assumption of imperfect reciprocity may perhaps be more realistic in practice. Nevertheless we also note that  our coding technique can also be applied to the case of perfect reciprocity, and improves upon the message transmission scheme~\cite{Lai2012}.

In other related works, an information theoretic framework for secure communication was pioneered by Shannon~\cite{shannon:49}. The problem of secret-key 
generation from common randomness between two legitimate terminals was introduced in~\cite{maurer:93,ahlswedeCsiszar:93}.  The setup considered
in these papers  involves one-way communication channel models and a public-discussion channel.  In followup works,  secret-key agreement capacity  over one-way fading channels with a noiseless public feedback link have been studied in e.g.~\cite{Agrawal2011,keys-arq,polar-codes-sec,arq-2,tang-liu}.
References~\cite{Khisti-For2011,Khisti-Globecom2011,khisti:10} study secret-key agreement over a channel with two-sided state sequence, which has application to  fading
channels. However these works do not incorporate the cost of acquiring channel state information which is of critical importance in the proposed setup.

\section{System Model and Main Results}
\label{sec:model}

\begin{figure}
\centering
\includegraphics[scale=0.7]{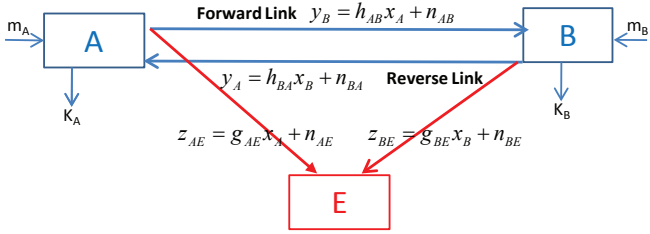}
\caption{Problem Setup. There are two legitimate terminals and one eavesdropper terminal. The legitimate terminals $A$ and $B$
 are required to agree on a common secret-key. The fading gains $\rvh_{AB}$ and $\rvh_{BA}$ are not revealed to any terminals, whereas
 the $\rvg_{AE}$ and $\rvg_{BE}$ are revealed to the eavesdropper.}
\label{fig:model}
\end{figure}
As illustrated in Fig.~\ref{fig:model}, we consider a setup with two legitimate terminals $A$ and $B$ and one eavesdropper $E$.
The terminals $A$ and $B$ communicate over a two-way non-coherent  wireless channel:
\allowdisplaybreaks{\begin{align}
& \rvy_B(t) \!=\! \rvh_{AB}(t) \rvx_A(t) \!+\! \rvn_{B}(t),~\rvy_A(t) \!=\! \rvh_{BA}(t) \rvx_B(t)\! +\! \rvn_{A}(t), \label{eq:model-2}\\
&\rvz_{AE}(t) \!\!=\!\! \rvg_{AE}(t) \rvx_A(t)\! +\! \rvn_{AE}(t),~
\rvz_{BE}(t)\!\! =\!\! \rvg_{BE}(t) \rvx_B(t) \!+\! \rvn_{BE}(t)\label{eq:model-4}
\end{align}}where $t\in\{1,\ldots, N\}$ denotes the time index, $\rvy_A(t)$ and $\rvy_B(t)$
denote the output symbols at terminals $A$ and $B$ and $\{\rvz_{AE}(t),\rvz_{BE}(t) \}$
denote the output symbols at terminal $E$.  The input symbols generated by terminals $A$
and $B$ at time $t$ are denoted by $\rvx_A(t)$ and $\rvx_B(t)$ respectively and are required to 
satisfy the average power constraints:
\begin{align*}
\frac{1}{N}\sum_{t=1}^N E[|\rvx_A(t)|^2] \le P, \quad \frac{1}{N}\sum_{t=1}^N E[|\rvx_B(t)|^2] \le P\end{align*}
The channel gains $(\rvh_{AB}, \rvh_{BA},\rvg_{AE}, \rvg_{BE})$ are drawn
from a  distribution $p({\rvh_{AB}, \rvh_{BA}})p(\rvg_{AE}, \rvg_{BE})$
once every $T$ consecutive symbols and stay constant over this period i.e., the
channel gains remain constant in the interval $t \in [iT+1, (i+1)T]$ for $i = 0,1,\ldots, \frac{N}{T}-1$.
We assume that the channel gains $\rvh_{AB}(t)$ and $\rvh_{BA}(t)$ are not revealed to any of the
terminals apriori, whereas the channel gains $(\rvg_{AE}(t), \rvg_{BE}(t))$ are revealed to the eavesdropper terminal.
The additive noise variables are drawn from an i.i.d.\ $\CN(0,1)$ distribution.
When computing the achievable rates, for convenience, we will assume that the channel gains are all  drawn from the Gaussian $\CN(0,1)$
distribution, although our results are not tied to this assumption. 
\begin{remark}
While we assume  that the channel gains of the eavesdropper are independent of $(\rvh_{AB},\rvh_{BA})$,
our results naturally extend to the case when the channel gains are drawn 
from a joint distribution $p(\rvh_{AB}, \rvh_{BA}, \rvg_{AE}, \rvg_{BE})$. 
Also note that the eavesdropper observes the transmission from the two terminals $A$ and $B$ over non-interfering links. 
An eavesdropper who observes a superposition of the two signals can only be weaker than the proposed eavesdropper.
\end{remark}

\begin{defn}[Secret-Key Capacity]
A feasible secret-key generation protocol is defined as follows. Terminals $A$ and $B$ sample independent random
variables $\rvm_A$ and $\rvm_B$ from a product distribution $p(\rvm_A)p(\rvm_B)$. At time $t,$ terminals
$A$ and $B$ generate the symbols $\rvx_A(t) = f_A(\rvm_A, \rvy_A^{t-1})$ and $\rvx_B(t) = f_B(\rvm_B, \rvy_B^{t-1})$.
At the end of $N$ channel uses, the terminals $A$ and $B$ generate secret keys $\rvk_A$ and $\rvk_B$ respectively
using the functions $\rvk_A = \cK_A(\rvy_A^N, \rvm_A)$ and $\rvk_B = \cK_B(\rvy_B^N, \rvm_B)$. We require that 
$\Pr(\rvk_A \neq \rvk_B) \le \eps_N$ and furthermore $\frac{1}{N}I(\rvk_A; \rvbz^N,\rvbg^K) \le \eps_N$ for some sequence
$\eps_N$ that goes to zero as $N \rightarrow\infty$. The largest achievable rate $R = \frac{1}{N}H(\rvk_A)$ is denoted as the
{\em secret-key capacity}. $\hfill\Box$
\end{defn}

As our main result, we establish upper and lower bounds on the secret-key capacity for the two-way non-coherent channel model. 

\subsection{Upper Bound}

\label{subsec-UB}
\begin{thm}
An upper-bound on the secret-key capacity of the two-way non-coherent  fading channel is, 
\begin{multline}
C \le R^+ = \frac{1}{T} I(\rvh_{AB}; \rvh_{BA})+ \\ \max_{\{P(\rvh_{AB})\} \in \cP_{AB}} \left\{I( \rvy_B; \rvx_A| \rvh_{AB}, \rvz_{AE}, \rvg_{AE})\right\}
\\ + \max_{\{P(\rvh_{BA})\} \in \cP_{BA}} \left\{{I( \rvy_A;\rvx_B|\rvh_{BA},\rvz_{BE}, \rvg_{BE}) }\right\}, \label{eq:R+}
\end{multline}
where the above expression is evaluated for \begin{align*}\rvx_A \sim \CN(0, P({\rvh_{AB}})),~
\rvx_B \sim\CN(0, P(\rvh_{BA})),\end{align*} and the set
$\cP_{AB}$ is the set of all non-negative power allocation functions $P({\rvh_{AB}})$ that satisfy the average power constraint
$E[P(\rvh_{AB})] \le P$ and the set $\cP_{BA}$ is defined similarly. 
\label{thm:ub}
\end{thm}

\begin{remark}
The upper bound expression in Theorem~\ref{thm:ub} has a natural interpretation. The term $\frac{1}{T}I(\rvh_{AB}; \rvh_{BA})$ is the contribution arising from reciprocal channel gains. The second term is the secret key rate achieved over a one way fading channel
\begin{align}
\rvy_B= \rvh_{AB}\rvx_A + \rvn_{B}, \quad \rvz_{AE}=\rvg_{AE}\rvx_A + \rvn_{AE}\label{eq:fading-oneway}
\end{align}
and a public discussion channel in the reverse link, when the fading chain $\rvh_{AB}$ is revealed to all the terminals
and the channel gain $\rvg_{AE}$ is revealed to the eavesdropper. Likewise the third term
is the contribution arising from the one-way fading channel in the reverse link. The upper bound expression~\eqref{eq:R+} indicates that the contribution of these three terms is additive. $\hfill\Box$
\end{remark}

The proof of Theorem~\ref{thm:ub} is provided in section~\ref{sec:ub}.

\subsection{Lower Bound - Public Discussion}
\begin{figure}
\includegraphics[scale=0.3]{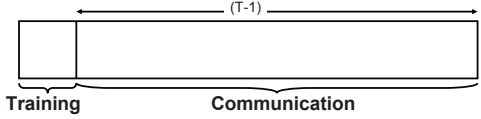}
\caption{Separation-based approach for secret-key generation. In each coherence block of length $T$ symbols, the first symbol is dedicated to training and a power $P_1$ is used i.e., $\rvx_A(iT+1)= \rvx_B(iT+1)=\sqrt{P_1}$. The remaining ${T-1}$ symbols are used for randomness generation. Each symbol is sampled i.i.d.\  $\CN(0,P_2)$ in this phase.}
\label{fig:separation}
\end{figure}

Our proposed coding scheme involves a separation based approach. As shown in Fig.~\ref{fig:separation}, the first symbol in each coherence
block is reserved for sending a training symbol and a total power of $P_1$ is used in this phase. This allows  the terminals  $B$ and $A$ to generate linear minimum
mean squared error estimates of the channel gains $\rvh_{AB}$ and $\rvh_{BA}$ respectively i.e.,
 \begin{align}
 \hh_{AB} = \al \rvh_{AB} + \rve_{AB} \label{eq:hh-AB} \\
 \hh_{BA} = \al \rvh_{BA} + \rve_{BA},\label{eq:hh-BA}
 \end{align}
 where $\al = \frac{P_1}{P_1 + 1}$ is the estimation coefficient, and $\rve_{AB}$ and $\rve_{BA}$ both have the distribution $\CN\left(0, \al(1-\al)\right)$ and are independent of the channel gains
 $\rvh_{AB}$ and $\rvh_{BA}$ respectively.
 
 For the reminder of coherence block, both the terminals transmit
 i.i.d.\ symbols from $\CN(0, P_2)$ to generate correlated source sequences i.e., in coherence period $i$ each component of the vector $\brvbx_A(i) \in {\mathbb C}^{T-1}$  is generated i.i.d.\ $\CN(0,P_2)$ and similarly each component of the vector $\brvbx_B(i)$ is i.i.d.\ $\CN(0,P_2)$.
 The corresponding received vectors are given by
 \begin{align}
 \brvby_A(i)\!\! =\!\! \rvh_{BA}(i) \brvbx_B(i)\! +\! \brvbn_{BA}(i), \brvby_B(i)\!\! =\!\! \rvh_{AB}(i)\brvbx_A(i)\! +\! \brvbn_{AB}(i).
 \end{align}
 At the end of $K$ such coherence blocks, as indicated in Table~\ref{tab:seq}, terminal $A$ has access to $(\hh_{BA}^K, \brvbx_A^K, \brvby_A^K)$ 
 whereas terminal $B$ has access to $(\hh_{AB}^K, \brvbx_B^K, \brvby_B^K)$. The eavesdropper observes $(\rvg_{AE}^K, \rvg_{BE}^K, \brvbz_{AE}^K, \brvbz_{BE}^K)$.
 \begin{table}[htdp]
\caption{Correlated Source and Channel Sequences at the Terminals}
\begin{center}
\begin{tabular}{|c|c|c|c|}\hline
 Sequence/Terminal &  A & B & E \\\hline
 Channel Sequence & $\hh_{BA}^K$ & $\hh_{AB}^K$ & $(\rvg_{AE}^K, \rvg_{BE}^K)$\\\hline
 Source Sequence - I & $\brvbx_A^K$ & $\brvby_B^K$ & $\brvbz_{AE}^K$ \\\hline
 Source Sequence - II & $\brvby_A^K$ & $\brvbx_B^K$ & $\brvbz_{BE}^K$\\\hline
\end{tabular}
\end{center}
\label{tab:seq}
\end{table}%

These correlated sequences are in turn used to generate a common secret-key by exchanging suitable public messages in the reconciliation phase~\cite{maurer:93,ahlswedeCsiszar:93}.   As indicated, in this section we assume that these public-messages for secret-key generation are transmitted over an external 
 public discussion channel. The resulting rate-expression is simpler and has a form that can be immediately compared with the upper bound.
 Subsequently we will remove this assumption and consider the transmission of public messages directly over the wireless channel.
\begin{thm}
\label{thm:disc-lb}
An achievable secret-key rate when an additional public discussion channel of arbitrarily high rate is available communication is given by
\begin{multline}
R^-_{PD} = \frac{1}{T}I(\hh_{AB}; \hh_{BA}) + \\
\frac{T-1}{T}\left\{I( \rvy_B;\rvx_A,\hh_{AB},\hh_{BA})-I(\rvy_B; \rvz_{AE}, \rvg_{AE},\rvh_{AB})\right\} + \\
\frac{T-1}{T}\left\{I(\rvy_A;\rvx_B,\hh_{BA},\hh_{BA})- I(\rvy_A;  \rvz_{BE}, \rvg_{BE},\rvh_{BA})\right\},\label{eq:lb-disc}
\end{multline}
where we evaluate the expression for $\rvx_A \sim \CN(0, P_2)$ and $\rvx_B \sim \CN(0, P_2)$ and $\hh_{AB}$ and $\hh_{BA}$ are specified in~\eqref{eq:hh-AB}
and~\eqref{eq:hh-BA} respectively and  $P_1$ and $P_2$  are non-negative and satisfy the relation 
\begin{align}
P_1 + (T-1) P_2 \le T P.
\label{eq:p-cons}
\end{align}
\end{thm}

A proof of Theorem~\ref{thm:disc-lb} is sketched in Section~\ref{sec:disc-lb}.

\begin{remark}
The  lower bound expression~\eqref{eq:lb-disc}  differs in the following respects from the upper bound expression~\eqref{eq:R+}:
(i) the channel gains $\rvh_{AB}$ and $\rvh_{BA}$ in the first term are replaced by their estimates $\hh_{AB}$ and $\hh_{BA}$ respectively;
(ii) the second and third terms are scaled by a factor of $\left(1-\frac{1}{T}\right)$; 
(iii) in computing the second and third terms, we assume that the legitimate receivers have access to the channel estimates whereas the eavesdropper has access to perfect channel states;
(iv) power allocation across the channel gains is not performed in the lower bound expression.
\end{remark}

An explicit evaluation of the lower bound
when that the channel gains $\rvh_{AB}$ and $\rvh_{BA}$ are jointly Gaussian, zero mean random variables with a cross-correlation of $\rho$
is as follows.

\begin{prop}
\label{prop:pdisc}
An achievable secret-key rate when terminals $A$ and $B$ have access to a public discussion channel is given by:
\begin{align}
R_{PD}^- = \frac{1}{T}R_{P,T}^- + \frac{T-1}{T}\left(R_{P, AB}^- + R_{P, BA}^-\right)\label{eq:pdisc}
\end{align}
where the expressions for $R_{P,T}^-,$ $R_{P,AB}^-$ and $R_{P, BA}^-$ are given by:
\begin{align}
R_{P,T}^- &= -\log\left(1-\al^2\rho^2\right) \label{eq:Rp_T}\\
R_{P, AB}^- &=  E\left[\log\left(\!\!1\!\! + \!\!\frac{P_2 |\rvh_{AB}|^2}{1+ P_2 |\rvg_{AE}|^2}\right)\right]  - \log\left(\!1 \!\!+\!\! \frac{P_2}{1+P_1}\!\right) \label{eq:Rp_AB}\\
R_{P,BA}^- &=E\left[\log\left(1\!\! +\!\! \frac{P_2 |\rvh_{BA}|^2}{1+ P_2 |\rvg_{BE}|^2}\right)\right]  - \log\left(\!1 \!\!+ \!\!\frac{P_2}{1+P_1}\!\right)\label{eq:Rp_BA}. 
\end{align} and $\al = \frac{P_1}{1+P_1}$. $\hfill\Box$
\end{prop}

The proof of Prop.~\ref{prop:pdisc} is provided in section~\ref{sec:pdisc}.

\subsection{Lower Bound - No Discussion}

\begin{figure}
\includegraphics[scale=0.31]{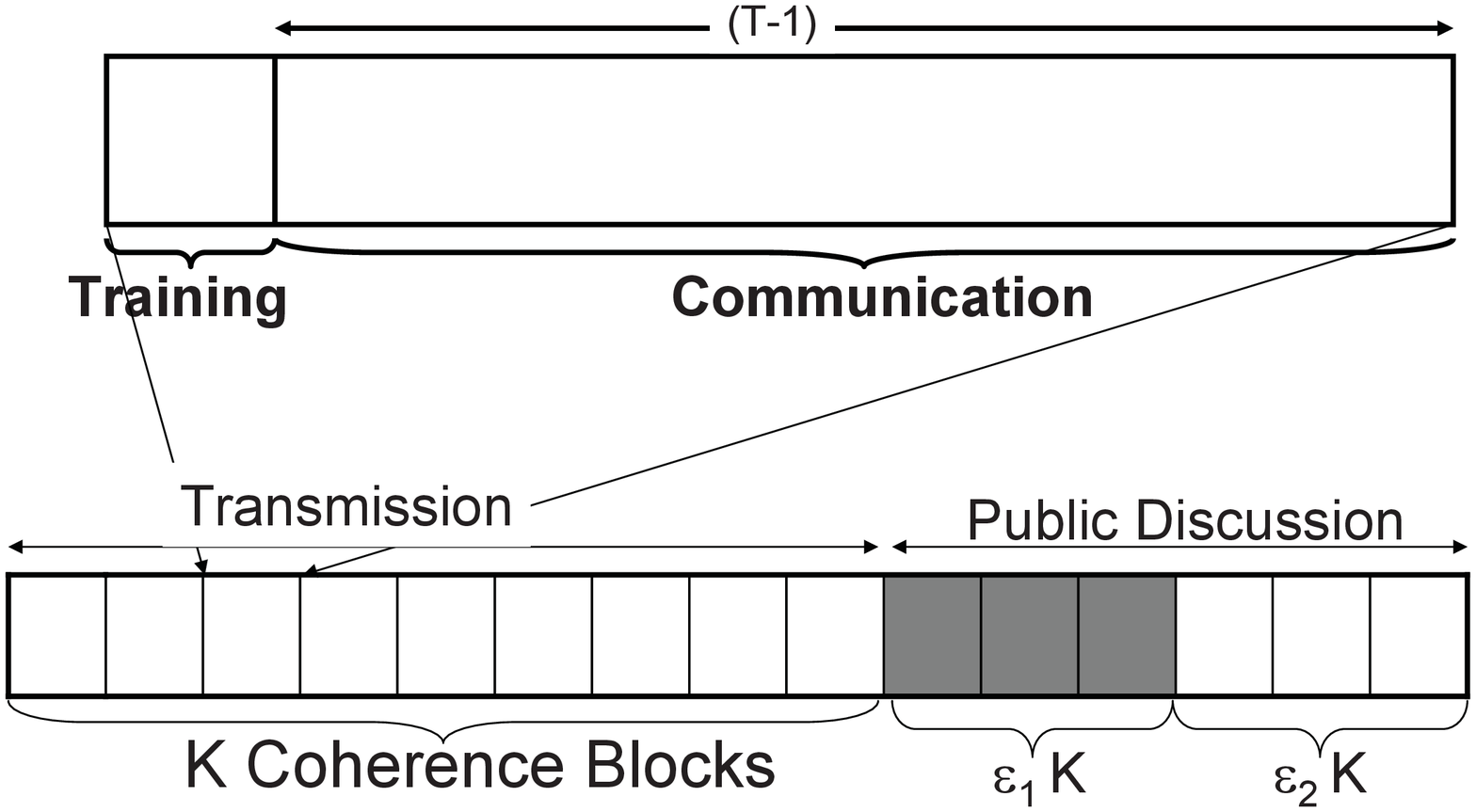}
\caption{Extension of the proposed coding scheme in absence of public discussion. A total of $K$ coherence blocks are used for training and source emulation. 
Thereafter a total of $\eps_1 K$ coherence blocks are used for transmission of the public message associated with the channel sequences and another $\eps_2 K$
coherence blocks are used for the transmission of the public message associated with the source sequences. }
\label{fig:no-disc}
\end{figure}

When a separate discussion channel is not available, we  use the communication channel for sending the public messages during the key-generation step as illustrated in Fig.~\ref{fig:no-disc}.  We also need to quantize the source sequences to satisfy the rate constraint imposed by the channel.

\begin{thm}
\label{thm:lb-nodisc}
An achievable secret-key rate in the absence of public discussion for the two-way reciprocal fading channel is:
\begin{align}
R^- = \frac{1}{1+\eps_1 + \eps_2} \left(\frac{1}{T}R_\mrm{I} + \frac{T-1}{T}R_\mrm{II}\right).
\end{align}
We provide the expressions for $R_\mrm{I}$ and $R_\mrm{II}$ below, both of which depend on $\eps_1$ and $\eps_2$. 
\begin{align}
R_\mrm{I} &= I(\rvu_{AB}; \hh_{BA}) \!\!+\!\! I(\rvu_{BA}; \hh_{AB}) \!\!-\!\! I(\rvu_{AB}; \rvu_{BA}) \label{eq:RI}
\end{align}
where the random variables $\rvu_{AB}$ and $\rvu_{BA}$ satisfy the Markov chain
\begin{align}
\rvu_{AB} \rightarrow \hh_{AB} \rightarrow \hh_{BA} \rightarrow \rvu_{BA} ,
\label{eq:u-markov}
\end{align}
and the rate constraints
\begin{align}
I(\rvu_{AB}; \hh_{AB}| \hh_{BA}) \le (T-1)\eps_1 R_\mrm{NC}(P), \label{eq:Rc-u1}\\
I(\rvu_{BA}; \hh_{BA}| \hh_{AB}) \le (T-1) \eps_1 R_\mrm{NC}(P),\label{eq:Rc-u2}
\end{align}
where $R_\mrm{NC}(P)$ is an achievable rate for the non-coherent block fading channel (see e.g.,~\cite{HassibiHoch:03}).

The rate expression for $R_\mrm{II}$ is expressed as
$R_\mrm{II} = R_\mrm{AB}^- + R_\mrm{BA}^-,$ where
\begin{align}
R_\mrm{AB}^- &= I(\rvv_B; \rvx_A, \rvu) - I(\rvv_B; \rvz_{AE}, \rvg_{AE},\rvh_{AB}),\label{eq:Rab} \\
R_\mrm{BA}^- &= I(\rvv_A; \rvx_B, \rvu) - I(\rvv_A; \rvz_{BE}, \rvg_{BE}, \rvh_{BA}) \label{eq:Rba}
\end{align}where the random variable $\rvu \defeq (\rvu_{AB}, \rvu_{BA})$. 
The random variables $\rvv_A$ and $\rvv_B$ in~\eqref{eq:Rba} and~\eqref{eq:Rab} satisfy the following
Markov Conditions
\begin{align}
\rvv_A \rightarrow   \rvy_A \rightarrow (\rvu,\rvx_B),~\rvv_B \rightarrow  \rvy_B \rightarrow (\rvu,\rvx_A).
\end{align}
as well as the rate constraints:
\begin{align}
I(\rvv_A; \rvy_A | \rvu, \rvx_B) \le \eps_2 R_\mrm{NC}(P), \label{eq:Rc-v1}\\
I(\rvv_B; \rvy_B | \rvu, \rvx_A) \le \eps_2 R_\mrm{NC}(P).\label{eq:Rc-v2}
\end{align}$\hfill\Box$
\end{thm}

A proof of Theorem~\ref{thm:lb-nodisc} appears in section~\ref{sec:nodisc}.

We further evaluate the rate expression in Theorem~\ref{thm:lb-nodisc} for a Gaussian test channel:
\begin{align}
\rvu_{AB} = \hh_{AB} + \rvq_{AB}, \quad \rvu_{BA} = \hh_{BA} + \rvq_{BA}\label{eq:u-test}
\end{align}
where $\rvq_{AB}, \rvq_{BA} \sim \CN(0, Q_1)$ are Gaussian random variables independent of all other variables.
Similarly we let
\begin{align}
\rvv_{A} = \rvy_A + \rvw_{A}, \quad  \rvv_B = \rvy_B+ \rvw_{B}\label{eq:v-test}
\end{align}
where $\rvw_A, \rvw_B \sim \CN(0,Q_2)$  are Gaussian random variables independent of all other variables. 

\begin{figure*}[htbp]
  \begin{minipage}[t]{0.49\linewidth}
    \centering
    \includegraphics[trim = 5mm 60mm 5mm 60mm, clip, width=\linewidth]{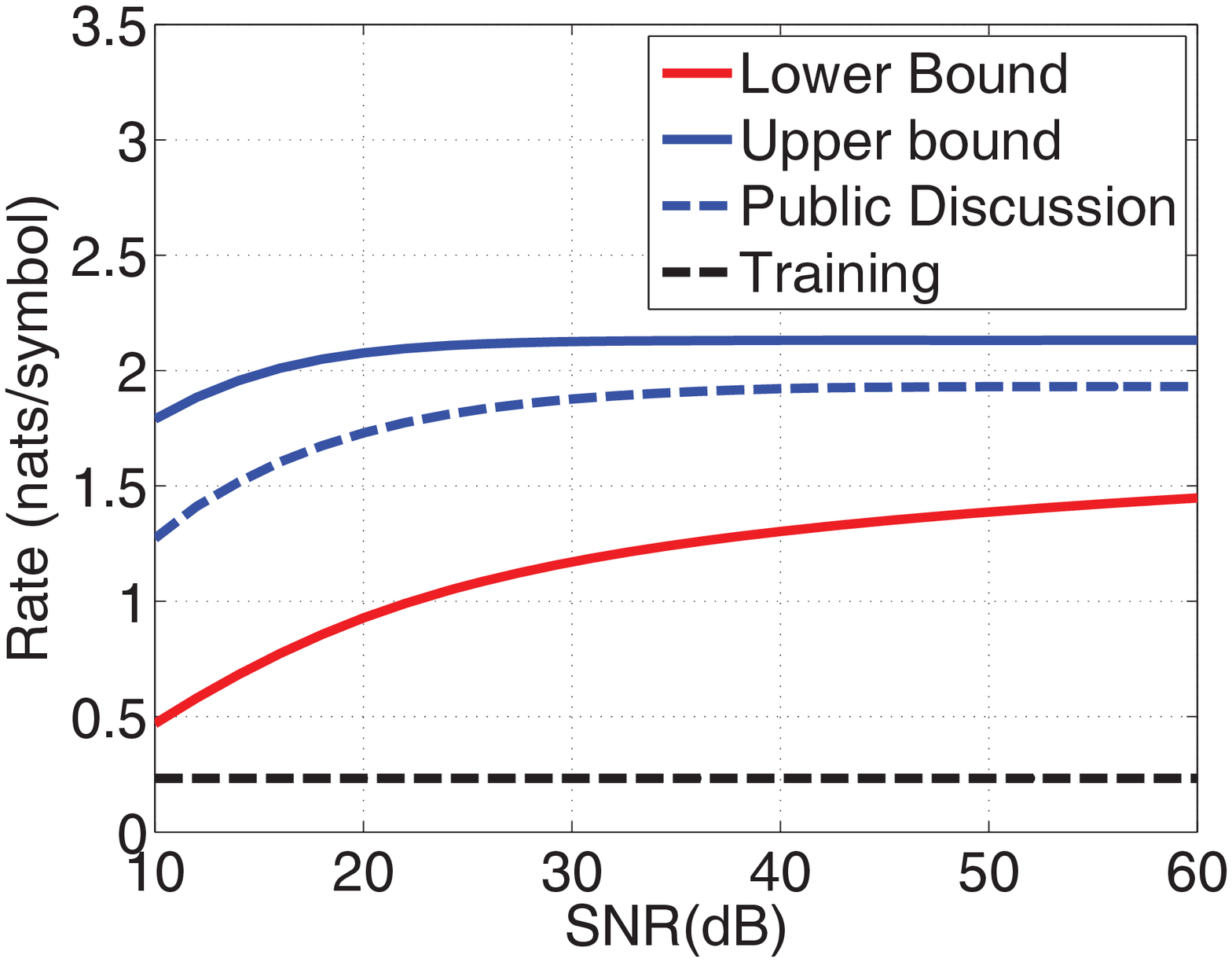}
    
    
    \caption{Bounds on the SK capacity as a function of SNR for a coherence period of $T=10$. }
    \label{fig:snr}
  \end{minipage}
  \hspace{0.5cm}
  \begin{minipage}[t]{0.49\linewidth}
    \centering
    \includegraphics[trim = 5mm 60mm 5mm 60mm, clip, width=\linewidth]{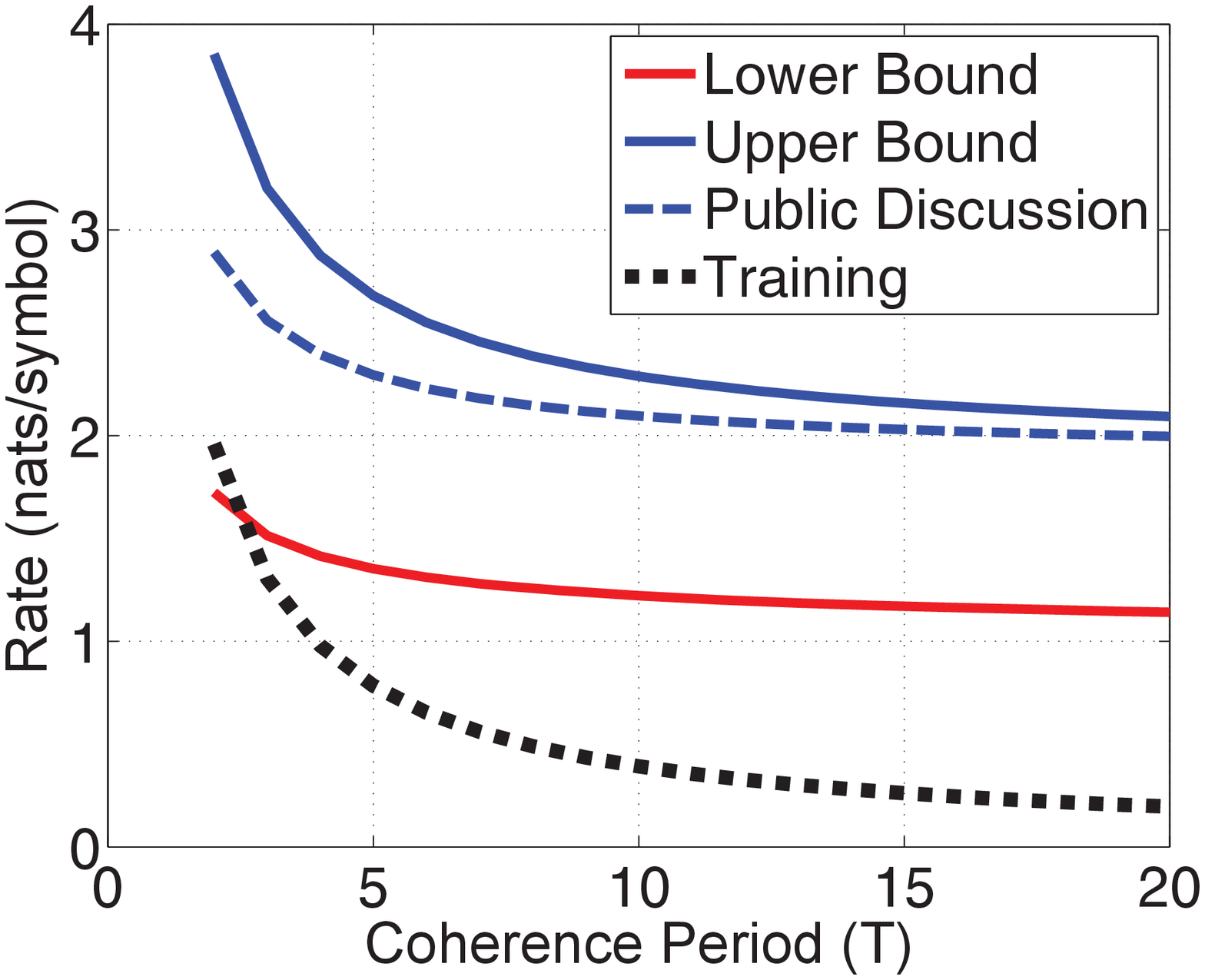}
        
    \caption{Bounds on the SK capacity as a function of coherence period for SNR = 30dB.}
    \label{fig:T}
  \end{minipage}
\end{figure*}


\begin{prop}
\label{thm:lb}
An achievable secret-key rate in the absence of public discussion for the two-way reciprocal fading channel is given by:
\begin{align}
\!\!R^-= \frac{1}{\!\!1\!\!+\!\!\eps_1 \!\!+\!\! \eps_2}\bigg\{ \frac{1}{T}R_\mrm{I}(\eps_1, P_1) + \frac{T-1}{T}R_\mrm{II}(\eps_2, P_2)\bigg\} \label{eq:R-}
\end{align}
where $P_1$ and $P_2$ are non-negative and satisfy~\eqref{eq:p-cons} and $\eps_1$ and $\eps_2$ are non-negative constants
that will be specified in the sequel. The rate  expressions
$R_\mrm{I}$ and $R_\mrm{II}$  are as follows.
 \begin{multline}
 R_\mrm{I}(\eps_1, P_1) =\\
\left\{\!\!-2\log\left(1-\frac{\al^2 \rho^2}{1+\frac{Q_1}{\al}}\right)\!\! +\!\! \log\left(1-\frac{\al^2 \rho^2}{\left(1+\frac{Q_1}{\al}\right)^2}\!\!\right)  \right\}\label{eq:RI-exp}
 \end{multline}
 where $\al = \frac{P_1}{1+P_1}$ and $Q_1$ satisfies
\begin{align}
\log\left(1 + \frac{\al(1-\al^2\rho^2)}{Q_1}\right) \le \eps_1 TR_\mrm{NC}(P). \label{eq:Q1-cons}
\end{align}
The expression for $R_\mrm{II}$ satisfies 
\begin{align}R_\mrm{II}(\eps_2, P_2) \!\!&= \!\!\bigg\{R_{AB}^-(\eps_2, P_2) \!\!+\!\! R_{BA}^-(\eps_2, P_2)\bigg\},\\
\!\!R_{AB}^-(\eps_2, P_2)\!\! &=\!\! E\left[\!\log\left(\!\!1+\!\! \frac{P_2 |\rvh_{AB}|^2}{(1 + Q_2 )(1+P_2 |\rvg_{AE}|^2\!)}\right)\!\right] \notag \\ &\qquad - \log\left(\frac{\sigma_{AB}^2 P_2}{1 + Q_2 } + 1\right)\label{eq:R-AB},\\
\!\!R_{BA}^-\!(\eps_2, P_2)\! &=\!\! E\left[\!\log\!\!\left(\!\!1\!\!+\!\! \frac{P_2 |\rvh_{BA}|^2}{(1 + Q_2 )(1+P_2 |\rvg_{BE}|^2)}\!\right)\!\right] \notag \\ &\qquad - \log\left(\frac{\sigma_{BA}^2 P_2}{1 + Q_2 } + 1\right),\label{eq:R-BA}
\end{align}
where \begin{align}
\sigma_{AB}^2 = \sigma_{BA}^2  = 1- \frac{\al^2}{Q_1 + \al}\label{eq:sigma}
\end{align} 
 and $Q_2$ satisfies
\begin{align}
\log\left(1 + \frac{\sigma_{AB}^2 P_2 + 1}{Q_2} \right)\le \eps_2 R_\mrm{NC}(P).\label{eq:Rb-cons}
\end{align}$\hfill\Box$
\end{prop}

A proof of Prop.~\ref{thm:lb} appears in section~\ref{subsec:nodisc-G}.

\begin{remark}
\label{rem:reduce}
The rate achieved in Prop.~\ref{thm:lb} reduces to the rate achieved using public discussion in Prop.~\ref{prop:pdisc} when we take ${Q_1=Q_2=0}$.
In particular when $Q_1=0$ note that the expression for $R_\mrm{I}$ in~\eqref{eq:RI-exp} immediately reduces to~\eqref{eq:pdisc}. Furthermore~\eqref{eq:sigma} reduces to
\begin{align}
\sigma_{AB}^2 = \sigma_{BA}^2  = \frac{1}{1+P_1}.\label{eq:sigma-2}
\end{align}
Substituting $Q_2=0,$ and~\eqref{eq:sigma-2} in~\eqref{eq:R-AB} and~\eqref{eq:R-BA} we obtain~\eqref{eq:Rp_AB} and~\eqref{eq:Rp_BA}
respectively. Thus the rate expressions in Prop.~\ref{thm:lb} are consistent with the expressions in Prop.~\ref{prop:pdisc} when  quantization (c.f.~\eqref{eq:u-test} and~\eqref{eq:v-test}) is not introduced.
 \end{remark}

We also observe that the upper and lower bounds are close in the high signal-to-noise-ratio (SNR) regime.

\begin{corol}
\label{corol:highSNR}
In the high SNR regime the upper and lower bounds satisfy the following relation:
\begin{align}
\lim_{P\rightarrow \infty} R^+(P) - R^-_\mrm{PD}(P) = \frac{1}{T}\g \label{eq:c1}\\
\lim_{P\rightarrow \infty} R^+(P) - R^-(P) = \frac{1}{T}\g\label{eq:c2}
\end{align}
where $R^+,$ $R^-_\mrm{PD}$ and $R^-$ are given by~\eqref{eq:R+},~\eqref{eq:pdisc} and~\eqref{eq:R-} respectively
and 
\begin{align}\g \defeq \!\!E\left[\log\left(1+ \frac{|\rvh_{AB}|^2}{|\rvg_{AE}|^2}\right)\right] \!\!+\!\! E\left[\log\left(1+ \frac{|\rvh_{BA}|^2}{|\rvg_{BE}|^2}\right)\right]\label{eq:g-def}\end{align}
is a constant that only depends on the distributions $p(\rvh_{AB})p(\rvg_{AE})$ and $p(\rvh_{BA})p(\rvg_{BE})$.
\end{corol}

A proof of Corollary~\ref{corol:highSNR} appears in section~\ref{sec:cap-highSNR}.


\subsection*{Numerical Comparison:}

Fig.~\ref{fig:snr} shows the bounds on secret-key capacity as a function of SNR when the coherence period $T=10$.
Fig.~\ref{fig:T} shows the bounds as a function of the coherence period $T$ when SNR = $35$ dB.
In Fig.~\ref{fig:snr} we fix the correlation coefficient in uplink and downlink gains to $\rho=0.95$ while in Fig.~\ref{fig:T}
it is fixed to $\rho = 0.99$. 

We make several observations in these plots. The lowermost plot in both figures corresponds to the best rate any training
based scheme can achieve i.e.,
\begin{align}
R_\mrm{training} = -\frac{1}{T}\log(1-\rho^2) \label{eq:R+train}
\end{align}
In computing~\eqref{eq:R+train} we assume that the legitimate receivers have a genie-aided side
information of the respective  channel gains and furthermore no overhead arising from the transmission
of public messages is considered. In Fig.~\ref{fig:snr} this rate is SNR independent and in Fig.~\ref{fig:T}
it decreases to zero as $\frac{1}{T}$.

The upper-most plot in Fig.~\ref{fig:snr} and Fig.~\ref{fig:T} is the upper bound on the secret-key agreement
capacity for any scheme. We note that for small values of $T$ the contribution from channel-reciprocity term in~\eqref{eq:R+}
is dominant. As $T$ increases the contribution of this term diminishes and the contribution from the communication
channel increases (c.f.~\eqref{eq:R+}). 

The remaining two plots are the lower bounds with and without public discussion. The lower bound with public discussion 
does not involve the overhead of transmitting a public message. For sufficiently high SNR, it approaches the upper bound.
The lower bound without public discussion does not come close to the upper in the SNR range of interest. Nevertheless
it achieves a much higher rate than the training based schemes.

\section{Proof of Theorem~\ref{thm:ub}}

\label{sec:ub}

In this section we provide a proof of Theorem~\ref{thm:ub}.  We assume that the communication 
is over a total of $K$ coherence blocks and let $N = T\cdot K$ denote the total number of channel uses.
Thus the channel state-sequences are of length $K$ whereas the channel input and output sequences are of length
$N$. We use the notation $\urvh_{AB}^K$ to denote the channel state sequence of length $K$ and the notation $\rvx_{A}^N$
to denote a channel input sequence of length $N$.
For the eavesdropper, we use the the notation $\rvbz^N \defeq (\rvz_{AE}^N,\rvz_{BE}^N)$
and  $\burvg^K \defeq ({\rvg}_{AE}^K,{\rvg}_{BE}^K)$.

From the Fano's inequality and the secrecy constraint we have that
\begin{align}
NR&\le I(\rvk_A; \rvk_B) - I( \rvk_A; \rvbz^N, \burvg^K)  + 2n\eps_n\\
&\le I(\rvk_A; \rvk_B| \rvbz^N, \burvg^K) + 2n\eps_n
\end{align}
where $\eps_n \rightarrow 0$ as $n\rightarrow \infty$. In the following steps the term $\eps_n$ will be suppressed. 
Since $\rvk_A = f_A(\rvm_A, \rvy_A^N)$ and $\rvk_B = f_B(\rvm_B, \rvy_B^N),$ 
using the data-processing inequality and the chain rule of mutual information, we have:
{{\begin{align}
&NR  \le 
I(\rvm_A,\urvh_{BA}^K, \rvy_A^N; \rvm_B, \urvh_{AB}^{K}, \rvy_B^N| \rvbz^N, \burvg^K)\notag\\
&=I(\rvm_A, \urvh_{BA}^K,\rvy_A^N; \rvm_B, \urvh_{AB}^K, \rvy_B^{N-1}| \rvbz^N, \burvg^K) +\notag\\ 
&I(\rvm_A, \urvh_{BA}^K, \rvy_A^N; \rvy_B(N)| \rvbz^N, \burvg^K,\rvm_B, \urvh_{AB}^K,\rvy_B^{N-1})\notag\\
&=I(\rvm_A, \urvh_{BA}^K,\rvy_A^{N-1}; \rvm_B, \urvh_{AB}^K, \rvy_B^{N-1}| \rvbz^N, \burvg^K)+\notag\\
&I(\rvy_A(N); \rvm_B, \urvh_{AB}^K, \rvy_B^{N-1}| \rvbz^N, \burvg^K, \rvm_A, \urvh_{BA}^K, \rvy_A^{N-1})+\notag\\
&I(\rvm_A, \urvh_{BA}^K, \rvy_A^{N-1}; \rvy_B(N)| \rvbz^N, \burvg^K,\rvm_B, \urvh_{AB}^K,\rvy_B^{N-1}) +\notag\\
&{{I(\rvy_A(N); \rvy_B(N)| \rvbz^N\!\!, \burvg^K,\!\!\rvm_B, \urvh_{AB}^K,\rvy_B^{N-1},\!\! \rvm_A, \urvh_{BA}^K, \!\!\rvy_A^{N-1})}}.
\label{eq:t-split}
\end{align}}}

\begin{figure*}{
\begin{align}\hline
&I(\rvm_A, \urvh_{BA}^K,\rvy_A^{N-1}; \rvm_B, \urvh_{AB}^K, \rvy_B^{N-1}| \rvbz^N, \rvbg^K)\label{eq:t1-start}\\
&\le I(\rvm_A, \urvh_{BA}^K,\rvy_A^{N-1}, \rvz_{AE}(N); \rvm_B, \urvh_{AB}^K, \rvy_B^{N-1}, \rvz_{BE}(N)| \rvbz^{N-1}, \burvg^K)\\
&= I(\rvm_A, \urvh_{BA}^K,\rvy_A^{N-1}; \rvm_B, \urvh_{AB}^K, \rvy_B^{N-1}, \rvz_{BE}(N)| \rvbz^{N-1}, \burvg^K)  \notag\\ &\qquad+ I(\rvz_{AE}(N); \rvm_B, \urvh_{AB}^K, \rvy_B^{N-1}, \rvz_{BE}(N)| \rvbz^{N-1}, \burvg^K,\rvm_A, \urvh_{BA}^K,\rvy_A^{N-1} )\\
&= I(\rvm_A, \urvh_{BA}^K,\rvy_A^{N-1}; \rvm_B, \urvh_{AB}^K, \rvy_B^{N-1}, \rvz_{BE}(N)| \rvbz^{N-1}, \burvg^K) \label{eq:markov-1}\\
&=I(\rvm_A, \urvh_{BA}^K,\rvy_A^{N-1}; \rvm_B, \urvh_{AB}^K, \rvy_B^{N-1}| \rvbz^{N-1}, \burvg^K) +I(\rvm_A, \urvh_{BA}^K,\rvy_A^{N-1}; \rvz_{BE}(N)| \rvbz^{N-1}, \burvg^K, \rvm_B, \urvh_{AB}^K, \rvy_B^{N-1})\\
&=I(\rvm_A, \urvh_{BA}^K,\rvy_A^{N-1}; \rvm_B, \urvh_{AB}^K, \rvy_B^{N-1}| \rvbz^{N-1}, \burvg^K)\label{eq:markov-2}\\\hline\notag
\end{align}}
\end{figure*}

We bound each of the four terms in~\eqref{eq:t-split}.  We show that the following condition holds
\begin{align}
&I(\rvm_A, \urvh_{BA}^K,\rvy_A^{N-1}; \rvm_B, \urvh_{AB}^K, \rvy_B^{N-1}| \rvbz^N, \burvg^K)\notag\\
&\le I(\rvm_A, \urvh_{BA}^K,\rvy_A^{N-1}; \rvm_B, \urvh_{AB}^K, \rvy_B^{N-1}| \rvbz^{N-1}, \burvg^K)\label{eq:ind-1}
\end{align}
in the steps between~\eqref{eq:t1-start} and~\eqref{eq:markov-2} on the top of  next page.
Note that~\eqref{eq:markov-1} follows from the fact that $\rvx_A(N)$ is a function of $(\rvm_A, \rvy_{A}^{N-1})$ and
furthermore  $\rvz_{AE}(N)$ is independent of all other random variables given $(\rvx_A(N), \rvg_{AE}(N))$ (c.f.~\eqref{eq:model-4}).
Eq.~\eqref{eq:markov-2} similarly follows because $\rvx_B(N)$ is a function of $(\rvm_B, \rvy_B^{N-1})$ and furthermore $\rvz_{BE}(N)$
is independent of all other random variables given $(\rvx_B(N), \rvg_{BE}(N))$. 

Furthermore, since $(\rvy_A(N),\rvz_{BE}(N))$ are independent of all other
random variables given $(\rvx_B(N), \rvh_{BA}(N), \rvg_{BE}(N))$ we can show that
\begin{align}
&I(\rvy_A(N); \rvm_B, \urvh_{AB}^K, \rvy_B^{N-1}| \rvbz^N, \burvg^K, \rvm_A, \urvh_{BA}^K, \rvy_A^{N-1})\notag\\
&\le I(\rvy_A(N); \rvx_B(N)| \rvz_{BE}(N), \rvg_{BE}(N),\rvh_{BA}(N)).\label{eq:ind-2}
\end{align}
and likewise
\begin{align}
&I(\rvm_A, \urvh_{BA}^K, \rvy_A^{N-1}; \rvy_B(N)| \rvbz^N, \burvg^K,\rvm_B, \urvh_{AB}^K,\rvy_B^{N-1})\notag \\
&\le I(\rvx_A(N); \rvy_B(N)| \rvz_{AE}(N), \rvh_{AB}(N), \rvg_{AE}(N))\label{eq:ind-3}.
\end{align}
Finally using the fact that $\rvx_A(N)$ and $\rvx_B(N)$ are functions of $(\rvm_A, \rvy_{BA}^{N-1})$ and $(\rvm_B, \rvy_{AB}^{N-1})$
respectively and the fact that $\rvy_A(N)$ is independent of all other random variables given $(\rvx_A(N),\rvh_{AB}(N))$
and similarly $\rvy_B(N)$ is independent of all other random variables given $(\rvx_B(N),\rvh_{BA}(N))$
 we can show that 
 \begin{align}
& I(\rvy_A(N); \rvy_B(N)| \rvbz^N\!\!, \burvg^K,\!\!\rvm_B, \urvh_{AB}^K,\rvy_B^{N-1},\!\! \rvm_A, \urvh_{BA}^K, \!\!\rvy_A^{N-1})\notag\\
 &\qquad=0.
 \end{align}Thus substituting~\eqref{eq:ind-1},~\eqref{eq:ind-2}
and~\eqref{eq:ind-3} into~\eqref{eq:t-split} we have that
\begin{align}
NR &\le I(\rvm_A, \urvh_{BA}^K,\rvy_A^{N-1}; \rvm_B, \urvh_{AB}^K, \rvy_B^{N-1}| \rvbz^{N-1}, \burvg^K)\notag\\
&+I(\rvy_A(N); \rvx_B(N)| \rvz_{BE}(N), \rvg_{BE}(N),\rvh_{BA}(N))\notag\\
&+I(\rvx_A(N); \rvy_B(N)| \rvz_{AE}(N), \rvh_{AB}(N), \rvg_{AE}(N))\label{eq:t-split-2}
\end{align}
Recursively applying the same steps we have that
\begin{align}
NR &\le I(\rvm_A, \urvh_{BA}^K; \urvh_{AB}^K, \rvm_B | \burvg^K) + \notag\\
&\sum_{i=1}^N I(\rvx_B(i); \rvy_A(i) | \rvz_{BE}(i), \rvg_{BE}(i), \rvh_{BA}(i))+ \notag\\
&\sum_{i=1}^N I(\rvx_A(i); \rvy_B(i) | \rvz_{AE}(i), \rvg_{AE}(i), \rvh_{AB}(i))\\
&= K I(\rvh_{BA}; \rvh_{AB}) +\notag\\
&\sum_{i=1}^N I(\rvx_B(i); \rvy_A(i) | \rvz_{BE}(i), \rvg_{BE}(i), \rvh_{BA}(i))+ \notag\\
&\sum_{i=1}^N I(\rvx_A(i); \rvy_B(i) | \rvz_{AE}(i), \rvg_{AE}(i), \rvh_{AB}(i))\label{eq:final-sum}
\end{align}
where in~\eqref{eq:final-sum} we use the fact that $(\rvm_A, \rvm_B)$ are mutually independent and independent
of the channel gains and furthermore $(\urvh_{AB}^K, \urvh_{BA}^K)$ are independent of $\burvg^K$.
Finally to establish the upper bound~\eqref{eq:R+} suppose that we assign a power $P_i(h_{AB})$
when the channel gain at time $i$ equals $h_{AB}$. Using the fact that a Gaussian input distribution
maximizes the conditional mutual information terms in~\eqref{eq:final-sum} (see e.g.,~\cite{khistiWornell:MISOME}) we have:
\begin{align}
& I(\rvx_A(i); \rvy_B(i) | \rvz_{AE}(i), \rvg_{AE}(i), \rvh_{AB}(i))  \notag\\
&\le E\left[\log\left(1 + \frac{P_i(\rvh_{AB})|\rvh_{AB}|^2}{1 + P_i(\rvh_{AB})|\rvg_{AE}|^2}\right)\right].   \label{eq:gauss-opt}   
\end{align}
Thus we have that
\begin{align}
&\sum_{i=1}^N I(\rvx_A(i); \rvy_B(i) | \rvz_{AE}(i), \rvg_{AE}(i), \rvh_{AB}(i))\notag\\
&\le\sum_{i=1}^N E\left[\log\left(1 + \frac{P_i(\rvh_{AB})|\rvh_{AB}|^2}{1 + P_i(\rvh_{AB})|\rvg_{AE}|^2}\right)\right]  \\
&= E\left[\sum_{i=1}^N \log\left(1 + \frac{P_i(\rvh_{AB})|\rvh_{AB}|^2}{1 + P_i(\rvh_{AB})|\rvg_{AE}|^2}\right)\right]\\
&\le N  E\left[ \log\left(1 + \frac{\frac{1}{N}\sum_{i=1}^N P_i(\rvh_{AB})|\rvh_{AB}|^2}{1 + \frac{1}{N}\sum_{i=1}^NP_i(\rvh_{AB})|\rvg_{AE}|^2}\right)\right]\label{eq:concave}\\
&= N  E\left[ \log\left(1 + \frac{P(\rvh_{AB})|\rvh_{AB}|^2}{1 + P(\rvh_{AB})|\rvg_{AE}|^2}\right)\right],\label{eq:pwr-1}\\
&= NI(\rvx_A; \rvy_B | \rvz_{AE}, \rvg_{AE}, \rvh_{AB})\label{eq:pwr-11}
\end{align}
where $P(\rvh_{AB}) \defeq \frac{1}{N} \sum_{i=1}^N P_i(\rvh_{AB})$ is the average power allocated when
the fading state equals $\rvh_{AB}$ and~\eqref{eq:concave} uses the fact that the function $f(x) = \log\left(1 + \frac{ax}{1+bx}\right)$
is a concave function in $x$ and hence Jensen's inequality applies. Also note that
\begin{align}
E[P(\rvh_{AB})] &= E\left[\frac{1}{N}\sum_{i=1}^N P_i(\rvh_{AB}) \right] \\
&= \frac{1}{N}\sum_{i=1}^N E[P_i(\rvh_{AB})]\\
&\le \frac{1}{N}\sum_{i=1}^N P_i \le P.
\end{align}
In a similar fashion we can show that
\begin{align}
&\sum_{i=1}^N I(\rvx_B(i); \rvy_A(i) | \rvz_{BE}(i), \rvg_{BE}(i), \rvh_{BA}(i)) \notag\\
&\le N  E\left[ \log\left(1 + \frac{P(\rvh_{BA})|\rvh_{BA}|^2}{1 + P(\rvh_{BA})|\rvg_{BE}|^2}\right)\right].\label{eq:pwr-2}\\
&= NI(\rvx_B; \rvy_A | \rvz_{BE}, \rvg_{BE}, \rvh_{BA})\label{eq:pwr-22}
\end{align}
The upper bound~\eqref{eq:R+} follows by substituting in~\eqref{eq:pwr-11} and~\eqref{eq:pwr-22} into~\eqref{eq:final-sum}.
This completes the proof of Theorem~\ref{thm:ub}.

\section{Proof of Theorem~\ref{thm:disc-lb}}

\label{sec:disc-lb}

We  sketch the coding scheme associated with Theorem~\ref{thm:disc-lb}, where we assume
the availability of a public discussion channel. In our proof  we assume that the channel
input and output symbols as well as the fading gains are discrete valued. We do not address the discretization of these
random variables in this section; it will be addressed in the proof of Theorem~\ref{thm:lb-nodisc},
when we do not assume the availability of a discussion channel.

\subsection{Generation of Correlated Source Sequences}
In each coherence block we reserve the first symbol for transmission of a pilot symbol
\begin{multline}
\rvx_A(iT+1) = \rvx_B(iT+1)  =\sqrt{P_1}, \\ i=0,\ldots, \frac{N}{T}-1.
\label{eq:training}
\end{multline}
and the remainder of the symbols for transmission of i.i.d.\  random variables i.e., 
\begin{multline}\rvx_A(t) \sim p_{\rvx_A}(\cdot),
\rvx_B(t) \sim p_{\rvx_B}(\cdot), \forall t \in [1, N],  t \neq iT+1.\label{eq:gauss}
\end{multline}

The legitimate receivers $B$ and $A$ use the corresponding output symbols $\rvy_B(iT+1)$
and $\rvy_A(iT+1)$ for estimating the underlying channel gains $\hh_{AB}(i+1)$ and $\hh_{BA}(i+1)$ 
in~\eqref{eq:hh-AB} and~\eqref{eq:hh-BA} respectively.

We use the bold-font notation $\brvbx_A(t)$ to denote the  sequence of ${T-1}$ i.i.d.\ symbols transmitted in coherence block $t$.
and let $\brvby_B(t)$ denote the corresponding output symbol vector in block $t$.  The source sequences at the terminals at the end of $K$ coherence blocks are indicated in Table~\ref{tab:seq}.

\subsection{Information Reconciliation}
Terminals $A$ and $B$ exchange public messages over the discussion channel to agree on identical sequence pair $(\hh_{AB}^K, \hh_{BA}^K, \brvby_{AB}^K, \brvby_{BA}^K)$. The public messages transmitted  are as follows
(i) Terminal $A$ bins the sequences $\{\hh_{BA}^K\}$ into $2^{K(H(\hh_{BA}|\hh_{AB} )+\eps)}$ bins and transmits the associated bin index $\phi_A$ over the discussion channel. 
(ii) Terminal $B$ bins the sequences $\{\hh_{AB}^K\}$ into $2^{K(H(\hh_{AB}|\hh_{BA})+\eps)}$ bins and transmits the associated  bin index $\phi_B$ over the discussion channel.
(iii) Terminal $A$ bins the sequences $\{\brvby_A^K\}$ into $2^{K(H(\brvby_A|\brvbx_B, \hh_{AB},\hh_{BA})+\eps)}$ and transmits the associated bin index $\psi_A$ over the discussion channel.
(iv) Terminal $B$ bins the sequence $\{\brvby_B^K\}$ into $2^{K(H(\brvby_B|\brvbx_A, \hh_{AB},\hh_{BA})+\eps)}$ and transmits the associated bin index $\psi_B$ over the discussion channel. 

Using standard arguments (see e.g.,~\cite[Chap.~10]{ElGamalKim}) Terminal $A$ can reconstruct the sequence pair $(\hh_{AB}^K, \brvby_B^K)$ (with high probability) given the messages $(\phi_B, \psi_B)$ and its side information  $(\hh_{BA}^K, \brvbx_A^K)$. Terminal $B$ can reconstruct the sequence pair $(\hh_{BA}^K, \brvby_A^K)$ (with high probability) given the messages $(\phi_A, \psi_A)$ and its side information sequence $(\hh_{AB}^K, \brvbx_B^K)$. Thus at the end of this phase both terminals have access to the common sequence pair $(\brvby_A^K, \brvby_B^K, \hh_{AB}^K, \hh_{BA}^K)$.  The eavesdropper has access to $(\brvbz^K, \rvbg^K,\phi_A, \phi_B, \psi_A, \psi_B)$.

\subsection{Equivocation Analysis}

 We lower bound the equivocation at the eavesdropper as follows:
\begin{align}
&\quad H(\brvby_A^K, \brvby_B^K,\hh_{AB}^K, \hh_{BA}^K|\brvbz^K, \rvbg^K,\phi_A, \phi_B, \psi_A, \psi_B)\notag\\
&=H(\brvby_A^K, \brvby_B^K,\hh_{AB}^K, \hh_{BA}^K|\brvbz^K, \rvbg^K) - H(\phi_A) - H(\phi_B)\notag\\
&\qquad - H(\psi_A) - H(\psi_B)\notag\\
&=H(\hh_{AB}^K, \hh_{BA}^K) + H(\brvby_A^K, \brvby_B^K|\brvbz^K, \rvbg^K,\hh_{AB}^K, \hh_{BA}^K) \notag\\ 
&\qquad - H(\phi_A) - H(\phi_B)- H(\psi_A) - H(\psi_B)\label{eq:hh-indep}\\
&\ge H(\hh_{AB}^K, \hh_{BA}^K) + H(\brvby_A^K, \brvby_B^K|\brvbz^K, \rvbg^K,\rvh_{AB}^K, \rvh_{BA}^K) \notag\\ 
&\qquad - H(\phi_A) - H(\phi_B)- H(\psi_A) - H(\psi_B)\label{eq:hh-cond}\\
&=\!\! H(\hh_{AB}^K, \hh_{BA}^K)\!\! +\!\! H(\brvby_A^K|\brvbz_B^K, \rvg_{BE}^K, \rvh_{BA}^K)\notag\\ &\qquad+\!\! H(\brvby_B^K|\brvbz_A^K, \rvg_{AE}^K, \rvh_{AB}^K)\notag\\ 
&\qquad - H(\phi_A) - H(\phi_B)- H(\psi_A) - H(\psi_B)\label{eq:y-cond}\\
&=KH(\hh_{AB},\hh_{BA}) + KH(\brvby_A|\brvbz_B, \rvg_{BE}, \rvh_{BA})\notag\\
&\qquad+ KH(\brvby_B|\brvbz_A, \rvg_{AE}, \rvh_{AB})\notag\\ 
&\qquad - H(\phi_A) - H(\phi_B)- H(\psi_A) - H(\psi_B)\label{eq:hh-mem}
\end{align}
where we use the chain rule of entropy and the fact that $(\hh_{AB}^K,\hh_{BA}^K)$ is independent of $(\brvbz^K, \rvbg^K)$ in~\eqref{eq:hh-indep};
in~\eqref{eq:hh-cond} we use the fact that conditioning reduces entropy as well as $\hh_{AB}^K$ and $\hh_{BA}^K$ are independent of the remaining variables gives $\rvh_{AB}^K$ and $\rvh_{BA}^K$ respectively; in~\eqref{eq:y-cond} we use the fact that $\brvbx_A^K$ and $\brvbx_B^K$ are sampled
independently and hence the following Markov conditions hold:
\begin{align}
&\brvby_A^K \rightarrow (\brvbz_B^K,\rvg_{BE}^K, \rvh_{BA}^K) \rightarrow (\brvby_B^K, \rvg_{AE}^K, \rvh_{AB}^K) \\
&\brvby_B^K \rightarrow (\brvbz_A^K, \rvg_{AE}^K, \rvh_{AB}^K) \rightarrow (\brvby_A^K, \rvg_{BE}^K, \rvh_{BA}^K);
\end{align}
Eq.~\eqref{eq:hh-mem} follows from the fact that the sequence pair $(\rvh_{AB}^K, \rvh_{BA}^K)$ is sampled i.i.d.\
and furthermore $\brvbx_A^K$ and $\brvbx_B^K$ are also sampled i.i.d.\ Furthermore,
\begin{align}
&KH(\rvh_{AB},\rvh_{BA}) - H(\phi_A) - H(\phi_B)\notag\\
&\ge KH(\hh_{AB},\hh_{BA}) - H(\hh_{BA}|\hh_{AB}) - H(\hh_{AB}|\hh_{BA})- 4K\eps\\
&= KI(\hh_{AB};\hh_{BA})-4K\eps.\label{eq:lb-ent-hh}
\end{align}
Using the fact that $\brvbx_A$ is sampled i.i.d. in~\eqref{eq:xa-iid} and letting $\eps' = \frac{\eps}{T-1}$ we have
\begin{align}
& K H(\brvby_A|\brvbz_B, \rvg_{BE}, \rvh_{BA}) - H(\psi_A) \notag\\
&\ge K H(\brvby_A|\brvbz_B, \rvg_{BE}, \rvh_{BA}) - KH(\brvby_A  | \brvbx_B, \hh_{AB}, \hh_{BA}) - 2K\eps\\
&\!\!=\!\! K(T-1)\left\{H(\rvy_A|\rvz_B, \rvg_{BE}, \rvh_{BA})\!\!-\!\! H(\rvy_A | \rvx_B, \hh_{AB}, \hh_{BA})\!\!- \!\!2\eps'\right\}\label{eq:xa-iid}\\
&= K(T-1)\!\!\left\{I(\rvy_A; \rvx_B, \hh_{AB}, \hh_{BA})\!\! -\!\! I(\rvy_A; \rvz_B, \rvg_{BE}, \rvh_{BA})\!\! -\!\!2\eps' \right\}.\label{eq:disc-sm-2}
\end{align}

In a similar fashion we can show that
\begin{align}
&K H(\brvby_B | \brvbz_A, \rvg_{AE},\rvh_{AB}) - H(\psi_B) \notag\\
&\ge K(T-1)\left\{I(\rvy_B; \rvx_A, \hh_{AB}, \hh_{BA})\!\!-\!\! I(\rvy_B; \rvz_A, \rvg_{AE}, \rvh_{AB})\!\!-\!\!2\eps'\right\}\label{eq:disc-sm-3}
\end{align}

Combining~\eqref{eq:lb-ent-hh},~\eqref{eq:disc-sm-2} and~\eqref{eq:disc-sm-3} we obtain that
\begin{align}
&H(\brvby_A^K, \brvby_B^K,\hh_{AB}^K, \hh_{BA}^K|\brvbz^K, \rvbg^K,\phi_A, \phi_B, \psi_A, \psi_B)\notag\\
&\ge KI(\hh_{AB};\hh_{BA})- 4K\eps - 4K\eps'\notag\\
& + K(T-1)\!\!\left\{I(\rvy_A; \rvx_B, \hh_{AB}, \hh_{BA})\!\! -\!\! I(\rvy_A; \rvz_B, \rvg_{BE}, \rvh_{BA})\!\! -\!\!2\eps' \right\}\notag\\
& +K(T-1)\left\{I(\rvy_B; \rvx_A, \hh_{AB}, \hh_{BA})\!\!-\!\! I(\rvy_B; \rvz_A, \rvg_{AE}, \rvh_{AB})\!\!-\!\!2\eps'\right\}\label{eq:R-disc-lb-eps}
\end{align}

In the final step, the equivocation lower bound in~\eqref{eq:R-disc-lb-eps} can be used to generate a secret-key. The associated construction is discussed in section~\ref{subsec:sec-key-gen} when we consider the case without a public discussion channel. We omit the details as they are completely analogous.

\subsection{Proof of Prop.~\ref{prop:pdisc}}
\label{sec:pdisc}

We evaluate the rate expression~\eqref{eq:lb-disc} in Theorem~\ref{thm:disc-lb} by selecting $\rvx_A \sim \CN(0,P_2)$ and $\rvx_B \sim \CN(0,P_2)$. Note that the MMSE estimates of $\rvh_{AB}$ and $\rvh_{BA}$ can be expressed as
\begin{align}
\hh_{AB}  &= \al \rvh_{AB} + \rve_{AB}\\
\hh_{BA} &= \al \rvh_{BA} + \rve_{BA}
\end{align}
where $\al = \frac{P_1}{P_1+1}$ is the linear MMSE coefficient and $\rve_{AB}, \rve_{BA} \sim \CN(0, \al(1-\al))$ are mutually independent and independent of $(\rvh_{AB}, \rvh_{BA})$. 

Thus  $\hh_{AB}$ and $\hh_{BA}$ are jointly Gaussian random variables with $E[|\hh_{AB}|^2] = E[|\hh_{BA}|^2] = \al$ and $E[\hh_{AB} \hh_{BA}^\dagger] = \al^2E[\rvh_{AB}\rvh_{BA}^\dagger]$. Hence one can show that \begin{align}
I(\hh_{AB}; \hh_{BA}) &= -\log\left(1-\al^2\rho^2\right). \label{Rp-Training-Comp}
 \end{align}
To compute the remaining terms in~\eqref{eq:lb-disc}
note that
\begin{align}
&I(\rvy_A; \rvx_B, \hh_{AB},\hh_{BA})-I(\rvy_A; \rvz_B, \rvg_{BE}, \rvh_{BA})\notag\\
&\ge h(\rvy_A |\rvz_B, \rvg_{BE},  \rvh_{BA} ) - h(\rvy_A | \rvx_B, \hh_{AB}, \hh_{BA})\\
&\ge h(\rvy_A |\rvz_B, \rvg_{BE},  \rvh_{BA} ) - h(\rvy_A | \rvx_B,  \hh_{BA})\label{eq:tap1}
\end{align}

The since  $\rvx_B \sim \CN(0, P_2),$ the first-term can be evaluated as:
\begin{align}
h(\rvy_A |\rvz_B, \rvg_{BE},  \rvh_{BA} ) &= E\left[\log\!2\pi e\left(1 +\frac{P_2 |\rvh_{BA}|^2}{1 + P_2|\rvg_{BE}|^2}\right)\right].\label{eq:tap2}
\end{align}

To evaluate the second term we introduce $\tilde{\rvh}_{BA} \defeq \rvh_{BA}-\hat{\rvh}_{BA}$. Note that since $\hh_{BA}$ is the MMSE estimate of $\rvh_{BA}$, we have that $\tilde{\rvh}_{BA} \sim \CN(0, \frac{1}{1+P_1})$
and hence
\begin{align}
h(\rvy_A | \rvx_B,  \hh_{BA}) &\le h(\rvy_A - \hh_{BA} \rvx_B)\\
&= h(\tilde{\rvh}_{BA} \rvx_B + \rvz_B)\label{eq:tap3a}\\
&\le \log2\pi e\left(1 + \frac{P_2}{1+P_1}\right)\label{eq:tap3}
\end{align}
where we use the fact that conditioning reduces the differential entropy in~\eqref{eq:tap3a}
and that  Gaussian distribution maximizes  entropy among all distributions with a fixed
variance in~\eqref{eq:tap3}.

From~\eqref{eq:tap1},~\eqref{eq:tap2} and~\eqref{eq:tap3} we have that
\begin{align}
R_\mrm{P, BA}^- &= I(\rvy_A; \rvx_B, \hh_{AB},\hh_{BA})-I(\rvy_A; \rvz_B, \rvg_{BE}, \rvh_{BA}) \\
&=E\left[\log\left(1 +\frac{P_2 |\rvh_{BA}|^2}{1 + P_2|\rvg_{BE}|^2}\right)\right] -\log\left(1 + \frac{P_2}{1+P_1}\right) 
\label{eq:Rp-BA-Comp}
\end{align}
which establishes~\eqref{eq:Rp_BA}. The expression~\eqref{eq:Rp_AB} for $R_\mrm{P,AB}^-$ can be established in a similar manner.

\section{Proof of Theorem~\ref{thm:lb-nodisc}}

\label{sec:nodisc}

\begin{center}
\begin{table*}[!htb]
\centering
\caption{Quantization Codebooks for Source and Channel Sequences}
\begin{tabular}{|c|c|c|c|} \hline
Codebook & Total Codewords & Total Bins & Codewords Per Bin \\\hline
Channel Quantization (Node A): $\cC_A^\rvu$ & $2^{K (I(\rvu_{BA}; \hh_{BA})+\eps)}$ & $2^{K (I(\rvu_{BA}; \hh_{BA}| \hh_{AB})+ 2\eps)}$ & $2^{K(I(\rvu_{BA}; \hh_{AB})-\eps)}$\\\hline
Channel Quantization (Node B): $\cC_B^\rvu$ & $2^{K(I(\rvu_{AB}; \hh_{AB})+\eps)}$  &$ 2^{K(I(\rvu_{AB}; \hh_{AB}| \hh_{BA})+2\eps)}$ & $2^{K(I(\rvu_{AB}; \hh_{BA})-\eps)}$\\\hline
Source Quantization (Node A): $\cC_A^\rvv$ & $2^{K(I(\rvbv_A;\rvby_A )+\eps)}$ & $2^{K(I(\rvbv_A; \rvby_A| \rvbx_B, \rvbu)+2\eps)}$&$2^{K(I(\rvbv_A; \rvbx_B, \rvbu)-\eps)}$\\\hline
Source Quantization (Node B): $\cC_B^\rvv$ & $2^{K(I(\rvbv_B; \rvby_B)+\eps)}$ & $2^{K(I(\rvbv_B;\rvby_B|\rvbx_A, \rvbu)+2\eps)}$ & $2^{K(I(\rvbv_B; \rvbx_A, \rvbu) -\eps)}$\\\hline
\end{tabular}
\label{tab:quant}
\end{table*}\end{center}

In absence of the discussion channel, we need to transmit the public messages during the error-reconciliation phase
using the same wireless channel. We again use $K$ coherence blocks to generate source and channel sequences as in
Table~\ref{tab:seq} using the training and communication phases as in~\eqref{eq:training} and~\eqref{eq:gauss}.
Thereafter we use ${\eps_1 \cdot K}$ coherence blocks for generating the secret-key from the channel sequences
and another ${\eps_2 \cdot K}$ coherence blocks for generating the secret-key from the source sequences. 
Due to this overhead, the total rate achieved is scaled by a factor of $\frac{1}{1+\eps_1 + \eps_2}$  
as in~\eqref{eq:R-}.

In the secret-key generation phase we suitably quantize each of the channel-state and source sequences to satisfy the rate constraint
for public messages. We describe the key steps of the coding scheme below.

\subsection{Quantization of Continuous Valued Random Variables}
While our lower bound is stated for continuous valued random variables, in the analysis of our coding scheme, we follow the approach in~\cite[pp.~50, sec. 3.4.1]{ElGamalKim} and assume that associated random 
variables are discrete valued.  In particular let $\rvx_B \sim p_{\rvx_B}(\cdot)$ be any continuous and bounded density function.
We consider an associated quantized version of $\rvx_B$, 
$[\rvx_{B}]_j \in  \{-j\Delta, -(j-1)\Delta, \ldots, (j-1)\Delta, j\Delta\}, $ with $\Delta = \frac{1}{\sqrt{j}}$, obtained by mapping $\rvx_B$
to the closest quantization point such that $|[\rvx_{B}]_j| \le |\rvx_B|$ holds.  Likewise for $\rvy_{A,j} = \rvh_{BA} [\rvx_{B}]_j+\rvn_A,$ 
let $[\rvy_{A,j}]_k$ be defined in a similar manner. The random variables $[\rvx_{A}]_j$ and $[\rvy_{B,j}]_k$ as well as the variables $[\rvu_{BA}]_k$, $[\rvu_{AB}]_k$, $[\rvv_A]_k$ and $[\rvv_B]_k$ associated with the quantization codebooks
are defined analogously. As we make the quantization finer and finer, the rate  using these
discrete random variables will approach the proposed lower bound with continuous valued random variables~\cite[pp.~65, sec.~3.8]{ElGamalKim}. In what follows, we will treat
these random variables as discrete valued and use the notion of robust-typicality~\cite[pp.~25-31, sec.~2.4]{ElGamalKim}.

\subsection{Codebook Construction}

Suitable quantization codebooks for the source and channel sequences are necessary in order to satisfy the rate-constraint associated with the wireless channel. Table~\ref{tab:quant} illustrates the different quantization codebooks used in our coding scheme, which are further described below.

\begin{itemize}
\item {\bf(Channel Quantization, Node A)}~$\cC_A^\rvu$:~We sample a total of ${M_A^\rvu = 2^{K (I(\rvu_{BA}; \hh_{BA})+\eps)}}$ codewords $\{\rvu_{BA}^K\}$ i.i.d.\ each of length $K$,
according to a distribution $p_{\rvu_{BA}}(\cdot)$. These codewords are divided into a total of ${T_A^\rvu = 2^{K (I(\rvu_{BA}; \hh_{BA}| \hh_{AB})+ 2\eps})}$
bins such that there are ${N_A^\rvu = 2^{K(I(\rvu_{BA}; \hh_{AB})-\eps)}} $ codewords per bin. 

\item {\bf(Channel Quantization, Node B)}~$\cC_B^\rvu$:~We sample a total of ${M_B^\rvu = 2^{K(I(\rvu_{AB}; \hh_{AB})+\eps)}}$ codewords $\{\rvu_{AB}^K\}$ i.i.d.\ according to a distribution $p_{\rvu_{AB}}(\cdot)$.
These codewords are divided into a total of ${T_B^\rvu = 2^{K(I(\rvu_{AB}; \hh_{AB}| \hh_{BA})+2\eps)}}$ bins such that there are ${N_B^\rvu = 2^{K(I(\rvu_{AB}; \hh_{BA})-\eps)}}$
codewords per bin. 

\item {\bf(Source Quantization, Node A)}~$\cC_A^\rvv$:~We sample a total of ${M_A^\rvv = 2^{K(I(\brvbv_A;\brvby_A )+\eps)}}$ codewords $\{\brvbv_A^K\}$ i.i.d.\ according to the distribution $p_{\brvbv_A}(\cdot)$. These codeword are partitioned
into a total of ${T_A^\rvv = 2^{K(I(\brvbv_A; \brvby_A| \brvbx_B, \rvbu)+2\eps)}}$ bins such that there are ${N_B^\rvv = 2^{K(I(\brvbv_A; \brvbx_B, \rvbu)-\eps)}}$ codewords per bin. 

\item {\bf(Source Quantization, Node B)}~$\cC_B^\rvv$:~We sample a total of ${M_B^\rvv = 2^{K(I(\brvbv_B; \brvby_B)+\eps)}}$ codewords $\{\brvbv_B^K\}$ i.i.d.\ according to the distribution $p_{\brvbv_B}(\cdot)$. These codewords are partitioned
into a total of ${T_B^\rvv = 2^{K(I(\brvbv_B;\brvby_B|\brvbx_A, \rvbu)+2\eps)}}$ bins such that there are ${N_B^\rvv = 2^{K(I(\brvbv_B; \brvbx_A, \rvbu) -\eps)}}$ codewords per bin.

\end{itemize}

\subsection{Encoding}
\begin{itemize}
\item Terminal $A$ finds a codeword $\rvu_{BA}^K \in \cC_A^\rvu$ that is jointly typical with $\hh_{BA}^K$. It finds the bin index $\phi_A$ of $\rvu_{BA}^K$ and transmits it as a public message over the channel. 
\item Terminal $B$ finds a codeword $\rvu_{AB}^K \in \cC_A^\rvu$ that is jointly typical with $\hh_{AB}^K$. It finds the bin index $\phi_B$ of  $\rvu_{AB}^K$ and transmits it as a public message over the channel.
\item Terminal $A$ finds a codeword $\brvbv_A^K \in \cC_A^\rvv$  that is jointly typical with $\brvby_A^K$. It finds the bin index $\psi_A$ of $\brvbv_A^K$ and transmits it as a public message over the channel.
\item Terminal $B$ finds a codeword $\brvbv_B^K \in \cC_B^\rvv$ that is jointly typical with $\brvby_B^K$. It finds the bin index $\psi_B$ of $\brvbv_B^K$ and transmits it as a public message over the channel.
\end{itemize}
We let $\cJ_1$ denote the error event that either terminal $A$ or terminal $B$ fails to find a codeword sequence typical with the source sequence and let $\cJ_2$ denote the error event associated with the transmission of the public message.

\subsection{Decoding}
Terminal $A$, upon receiving $\phi_B$, searches for a codeword $\rvu_{AB}^K$ in the bin-index $\phi_B\in \cC_B^\rvu$ that is jointly typical with $\hh_{BA}^K$. It declares an error if none or more than one sequence are jointly typical.  Terminal $A$, upon receiving $\psi_B$, searches for a codeword $\brvbv_B^K$ in the bin-index $\psi_B \in \cC_B^\rvv$ that is jointly typical with $(\brvby_A^K,\rvbu^K)$. It declares an error if none or more than one sequence appears to be typical.  We define $\cE_1 = \Pr(\hat{\rvu}_{AB}^K \neq \rvu_{AB}^K \cup \hat{\brvbv}_B^K \neq \brvbv_B^K)$.

Terminal $B$, upon receiving $\phi_A$, searches for a codeword $\rvu_{BA}^K$ in the bin-index $\phi_A \in \cC_A^\rvu$  that is jointly typical with $\hh_{AB}^K$. It declares an error if none or more than one sequences are jointly typical. Terminal $B$, upon receiving $\psi_A$, searches for a codeword $\brvbv_A^K$ in the bin-index $\psi_A \in \cC_A^\rvv$ that is jointly typical with $(\brvby_B^K,\rvbu^K)$. It declares
an error if none or more than one sequence appears to be typical. We define $\cE_2 = \Pr(\hat{\rvu}_{BA}^K \neq \rvu_{BA}^K \cup \hat{\brvbv}_A^K \neq \brvbv_A^K)$.

\subsection{Error Analysis (Sketch)}

Let ${\cE = \cE_1 \cup \cE_2 \cup \cJ_1 \cup \cJ_2}$ denote the union of all error events. It suffices to show that both $\cE_i$ and $\cJ_i$ have vanishing probabilities for $i\in \{1,2\}$.
Since the total number of codeword sequences in $\cC_A^\rvu$ equals $2^{K(I(\rvu_{BA}; \hh_{BA})+\eps)}$, it follows from the Covering Lemma~\cite[pp.~62, Lemma~3.3]{ElGamalKim}
that the probability that no jointly typical codeword exists goes to zero as $K\rightarrow \infty$. By a similar argument we have that $\Pr(\cJ_1) \rightarrow 0$.  Since ${H(\phi_A) \le K(I(\rvu_{BA}; \hh_{BA}|\hh_{AB})+2\eps)}$ and since the rate constraint~\eqref{eq:Rc-u2}  is imposed,  terminal $B$ decodes $\phi_A$ with high probability. In a similar fashion it can be argued that $\Pr(\cJ_2) \rightarrow 0$ as $K \rightarrow \infty$. Since the total number of codewords in each bin of $\cC_B^\rvu$ is  $2^{K(I(\rvu_{AB};\hh_{BA} )-\eps)}$
it follows from the Packing Lemma~\cite[pp.~46, Lemma~3.1]{ElGamalKim} that terminal $A$ will decode $\{\hat{\rvu}_{AB}^K = \rvu_{AB}^K\}$ with high probability. By analogous arguments it can be verified
that $\Pr(\cE_1) \rightarrow 0$ and $\Pr(\cE_2) \rightarrow 0$ as $K \rightarrow\infty$.

\subsection{Equivocation Analysis}
Let us define $\brvbz^K = (\brvbz_A^K, \brvbz_B^K)$ and $\rvbg^K  = (\rvg_{AE}^K, \rvg_{BE}^K)$. 
We first consider the conditional entropy term\footnote{All the expressions are conditioned on the codebook $\cC$. We suppress this
conditioning for sake of convenience.}
{\allowdisplaybreaks{\begin{align}
&H(\rvu_{BA}^K,\rvu_{AB}^K, \brvbv_A^K, \brvbv_B^K | \phi_A, \phi_B, \psi_A, \psi_B, \brvbz^K, \rvbg^K) \ge \notag\\
&\ge H(\rvu_{BA}^K, \rvu_{AB}^K, \brvbv_A^K, \brvbv_B^K | \brvbz^K, \rvbg^K) - H(\phi_A) \notag\\ &\quad - H(\phi_B)-H(\psi_A) -H(\psi_B) \\
&= H(\rvu_{BA}^K, \rvu_{AB}^K | \brvbz^K, \rvbg^K) \notag\\ &\qquad + H(\brvbv_A^K, \brvbv_B^K | \rvu_{BA}^K, \rvu_{AB}^K, \brvbz^K, \rvbg^K) \notag\\
&\quad - H(\phi_A)  - H(\phi_B)-H(\psi_A) -H(\psi_B) \\
&=H(\rvu_{BA}^K, \rvu_{AB}^K ) + H(\brvbv_A^K, \brvbv_B^K | \rvu_{BA}^K, \rvu_{AB}^K, \brvbz^K, \rvbg^K) \notag\\
&\quad - H(\phi_A)  - H(\phi_B)-H(\psi_A) -H(\psi_B) \label{eq:u-x-indep} 
\end{align}}}
where~\eqref{eq:u-x-indep} follows from the fact that $(\rvu_{BA}^K, \rvu_{AB}^K)$ are independent of $(\brvbx_A^K, \brvbx_B^K)$
and hence $(\brvbz^K, \rvbg^K)$.

From the construction recall that $\phi_A$, $\phi_B$, $\psi_A$ and $\psi_B$ denote the bin indices in $\cC_A^\rvu$, 
 $\cC_B^\rvu$, $\cC_A^\rvv$ and $\cC_B^\rvv$ respectively. Therefore
\begin{align}
&H(\phi_A) \le I(\rvu_{BA};\hh_{BA}|\hh_{AB})+2\eps,\label{eq:bin-bnd-1}\\
&H(\phi_B) \le I(\rvu_{AB}; \hh_{AB}|\hh_{BA})+2\eps,\\
&H(\psi_A) \le I(\brvbv_A; \brvby_A|\brvbx_B, \rvbu )+2\eps,\\
&H(\psi_B) \le I(\brvbv_B; \brvby_B|\brvbx_A, \rvbu)+2\eps.\label{eq:bin-bnd-2}
\end{align}
Thus we need to lower bound the joint-entropy  terms in~\eqref{eq:u-x-indep}.
We use the following Lemma.

\begin{lemma}
\label{lem:cond-ent}
Consider a triplet of random variables $(\rvu,\rvx,\rvy)$ that satisfy $\rvu\rightarrow\rvx\rightarrow\rvy$. 
Suppose that the sequence $\rvx^n$ is sampled i.i.d.\ $p_{\rvx}(\cdot)$.  Suppose that we generate a set $\cA$ 
consisting of $M=2^{nR}$ sequences $\{\rvu^n(i)\}_{1\le i \le M}$
each sampled i.i.d.\ from a distribution $p_{\rvu}(\cdot)$. Given $\rvx^n$ we select an index $L$ such that $(\rvx^n, \rvu^n(L))\in \cT_\eps^n(\rvx,\rvu)$.
If no such index exists we select $L$ uniformly at random.
If $\rvy^n$ is sampled i.i.d. from the conditional distribution $p_{\rvy|\rvx}(\cdot|\cdot)$, then
\begin{align}
&H(\rvy^n|\rvu^n(L)) \ge n\left\{H(\rvy|\rvu) + I(\rvx;\rvu)-R-\g_n(\eps))\right\}\label{eq:cond-ent}
\end{align}for any $R > I(\rvx;\rvu)$ and for some $\g_n(\eps)$ that goes to zero as $\eps \rightarrow 0$ and $n\rightarrow\infty$.
\end{lemma}

\begin{proof}
See Appendix~\ref{app:cond-ent}
\end{proof}

We first lower bound the term $H(\rvu_{BA}^K, \rvu_{AB}^K)$. Notice that $\rvu_{BA}^K \rightarrow \hh_{BA}^K \rightarrow \rvu_{AB}^K$ holds.
Since  $\rvu_{BA}^K$ and $\rvu_{AB}^K$ are distributed, nearly uniformly on the codebooks $\cC_A^\rvu$ and $\cC_B^\rvu$ respectively
(see e.g.,~\cite[pp.~660, Lemma 1]{khistiDiggaviWornell:12}) we have that
\begin{align}
H(\rvu_{BA}^K) &\ge K(I(\rvu_{BA}; \hh_{BA}) - \g_K(\eps))\\
H(\rvu_{AB}^K) &\ge K(I(\rvu_{AB}; \hh_{AB})-\g_K(\eps))\label{eq:ub-bnd}
\end{align}
Substituting $\rvx = \hh_{AB}$, $\rvy = \hh_{BA}$ and $\rvu = \rvu_{BA}$ and $R=I(\rvu_{BA};\hh_{AB})+2\eps$ in Lemma~\ref{lem:cond-ent} we have that
\begin{align}
H(\hh_{BA}^K|\rvu_{BA}^K) \ge K\left\{H(\hh_{BA}|\rvu_{BA}) -\g_K(\eps)\right\}.\label{eq:hba-bnd-0}
\end{align}

We define the indicator function $\mathbb{I}_\eps(\hh_{BA}^K, \rvu_{AB}^K, \rvu_{BA}^K)$ to equal $1$ if $(\hh_{BA}^K, \rvu_{AB}^K, \rvu_{BA}^K) \in \cT_\eps^K$
and zero otherwise.
Using the Markov Lemma~\cite[sec 12.1.1]{ElGamalKim} we have that $\Pr(\mathbb{I}(\hh_{BA}^K, \rvu_{AB}^K, \rvu_{BA}^K)=1) \rightarrow 1$ 
as $K\rightarrow\infty$ for every $\eps>0$. Hence
\begin{align}
&H(\hh_{BA}^K | \rvu_{BA}^K, \rvu_{AB}^K) \le H(\hh_{BA}^K | \rvu_{BA}^K, \rvu_{AB}^K, \mathbb{I}_\eps)+1\\
&\le H(\hh_{BA}^K | \rvu_{BA}^K, \rvu_{AB}^K, \mathbb{I}_\eps=1)\Pr(\mathbb{I}_\eps=1)+\notag\\
&\quad H(\hh_{BA}^K | \rvu_{BA}^K, \rvu_{AB}^K, \mathbb{I}_\eps=0)\Pr(\mathbb{I}_\eps=0)+1\label{eq:hba-bnd-t1}
\end{align}
Now consider
\begin{align}
&H(\hh_{BA}^K | \rvu_{BA}^K, \rvu_{AB}^K, \mathbb{I}_\eps=1)
\le E\left[\log |\cT_\eps^K(\hh_{BA}| \rvu^K_{A}, \rvu^K_B)|\right]\notag\\
&\le K (H(\hh_{BA}|\rvu_{BA}, \rvu_{AB}) + \g_K(\eps))\label{eq:typ-set-bound}
\end{align}where the last step follows from the  bound on the size of a conditionally typical set~\cite[pp.~27]{ElGamalKim}.

Furthermore since $\Pr({\mathbb I}_\eps=0) \rightarrow 0$ as $K\rightarrow \infty$ we have that~\eqref{eq:hba-bnd-t1} simplifies to the following
\begin{align}
H(\hh_{BA}^K | \rvu_{BA}^K, \rvu_{AB}^K) \le K\left\{H(\hh_{BA}|\rvu_{BA}, \rvu_{AB}) +\g_K(\eps) \right\}\label{eq:hba-bnd-1}
\end{align}

We have the following lower bound on $H(\rvu_{BA}^K, \rvu_{AB}^K)$:
\begin{align}
&H(\rvu_{BA}^K, \rvu_{AB}^K) = H(\rvu_{BA}^K | \rvu_{AB}^K) + H(\rvu_{AB}^K)\\
&= H(\rvu_{BA}^K, \hh_{BA}^K|\rvu_{AB}^K) - H(\hh_{BA}^K | \rvu_{BA}^K, \rvu_{AB}^K) + H(\rvu_{AB}^K)\\
&\ge H(\hh_{BA}^K | \rvu_{AB}^K)- H(\hh_{BA}^K | \rvu_{BA}^K, \rvu_{AB}^K) + H(\rvu_{AB}^K)\\
&\ge K\left\{I(\hh_{BA}; \rvu_{BA} | \rvu_{AB}) + I(\hh_{AB}; \rvu_{AB}) - \g_K(\eps)\right\}\label{eq:mut-inf-bnd}\\
&= K\left\{I(\hh_{BA}; \rvu_{BA})\!\! +\!\! I(\hh_{AB}; \rvu_{AB})\!\!-\!\!I(\rvu_{BA};\rvu_{AB})\!\!-\!\!3\g_K(\eps)\right\}\label{eq:uab-bnd}
\end{align}where we use~\eqref{eq:hba-bnd-0} and~\eqref{eq:hba-bnd-1} in~\eqref{eq:mut-inf-bnd}.

Now consider 
\begin{align}
&H(\brvbv_A^K, \brvbv_B^K | \rvu_{BA}^K, \rvu_{AB}^K, \brvbz^K, \rvbg^K) \notag\\
&\ge H(\brvbv_A^K, \brvbv_B^K | \rvu_{BA}^K, \rvu_{AB}^K, \brvbz^K, \rvbg^K, \rvh_{AB}^K, \rvh_{BA}^K) \label{eq:cond-t1}\\ 
&=  H(\brvbv_A^K | \rvh_{BA}^K, \brvbz_B^K, \rvg_{BE}^K)+ H(\brvbv_B^K | \rvh_{AB}^K, \brvbz_A^K, \rvg_{AE}^K)\label{eq:mut-indep-xaxb}
\end{align}
where~\eqref{eq:cond-t1} follows from the fact that conditioning reduces entropy and~\eqref{eq:mut-indep-xaxb} follows from the
fact that $(\brvbx_A^K, \brvbx_B^K)$ are generated mutually independently and hence the associated Markov chain holds.
We lower bound each of the two terms in~\eqref{eq:mut-indep-xaxb}
\begin{align}
&H(\brvbv_A^K | \rvh_{BA}^K, \brvbz_B^K, \rvg_{BE}^K)\notag\\ &= H(\rvh_{BA}^K, \brvbz_B^K, \rvg_{BE}^K |\brvbv_A^K) + H(\brvbv_A^K) - H( \rvh_{BA}^K, \brvbz_B^K, \rvg_{BE}^K)\notag\\
&=H(\rvh_{BA}^K, \brvbz_B^K, \rvg_{BE}^K |\brvbv_A^K) + H(\brvbv_A^K) - KH( \rvh_{BA}, \brvbz_B, \rvg_{BE}) \label{eq:src-iid}\\
&\ge H(\rvh_{BA}^K, \brvbz_B^K, \rvg_{BE}^K |\brvbv_A^K) + \notag\\ &\qquad KI(\brvbv_A;\brvby_A) - KH( \rvh_{BA}, \rvbz_B, \rvg_{BE}) -K\g_K(\eps)\label{eq:va-bnd}\\
&\ge K\bigg\{H(\rvh_{BA},\brvbz_B,\rvg_{BE}|\brvbv_A ) + \notag\\ &\qquad I(\brvbv_A;\brvby_A) - H( \rvh_{BA}, \brvbz_B, \rvg_{BE}) -\g_K(\eps) \bigg\}\label{eq:va-bnd-1}\\
&=K\left\{I(\brvbv_A;\brvby_A)-I(\brvbv_A;\brvbz_B,\rvh_{BA}, \rvg_{BE})-\g_K(\eps)\right\}\label{eq:va-bnd-2}
\end{align}where~\eqref{eq:src-iid} follows from the fact that $(\rvh_{AB}^K,\brvbz_B^K,\rvg_{BE}^K)$ are sampled i.i.d.\ 
and~\eqref{eq:va-bnd} follows from the fact that $\brvbv_A^K$ is distributed nearly uniformly over the set $\cC_A^\rvv$
(see e.g.,~\cite[pp.~660, Lemma 1]{khistiDiggaviWornell:12}) and~\eqref{eq:va-bnd-1} follows from Lemma~\ref{lem:cond-ent} by substituting
${\rvu=\brvbv_A}$, ${\rvy = (\rvh_{BA},\brvbz_B,\rvg_{BE})}$, ${\rvx=\rvby_A}$ and ${R = I(\brvby_A;\brvbv_A) +2\eps}$.
In a similar manner we can show that
\begin{align}
&H(\brvbv_B^K | \rvh_{AB}^K, \brvbz_A^K, \rvg_{AE}^K) \notag\\ &\ge K\left\{I(\brvbv_B;\brvby_B)-I(\brvbv_B;\brvbz_A,\rvh_{AB},\rvg_{AE})-\g_K(\eps)\right\} \label{eq:vb-bnd-2}
\end{align}

Substituting~\eqref{eq:bin-bnd-1}-\eqref{eq:bin-bnd-2},~\eqref{eq:uab-bnd},~\eqref{eq:va-bnd-2} and~\eqref{eq:vb-bnd-2} into
~\eqref{eq:u-x-indep} we have established that
\begin{align}
&\frac{1}{K}H(\rvu_{BA}^K,\rvu_{AB}^K, \brvbv_A^K, \brvbv_B^K | \phi_A, \phi_B, \psi_A, \psi_B, \brvbz^K, \rvbg^K) \ge \notag\\
&I(\hh_{AB}; \rvu_{BA}) + I(\hh_{BA}; \rvu_{AB})-I(\rvu_{BA};\rvu_{AB}) + \notag\\
&I(\brvbv_B; \brvbx_A, \rvbu)  - I(\brvbv_B; \rvh_{AB},\rvg_{AE}, \brvbz_A) + \notag\\
&I(\brvbv_A;\brvbx_B,\rvbu)  - I(\brvbv_A; \rvh_{BA},\rvg_{BE}, \brvbz_B)-\g_K(\eps).\label{eq:lb-t1}
\end{align}

Now observe that
\begin{align}
&I(\brvbv_B; \brvbx_A, \rvbu)  - I(\brvbv_B; \rvh_{AB},\rvg_{AE}, \brvbz_A)  \notag\\
&=I(\brvbv_B; \brvbx_A, \rvbu)  - I(\brvbv_B; \rvh_{AB},\rvg_{AE}, \brvbz_A,\rvbu) \label{eq:u-markov1}\\
&=H(\brvbv_B | \rvh_{AB},\rvg_{AE}, \brvbz_A,\rvbu) - H(\brvbv_B | \brvbx_A, \rvbu)\notag\\
&=(T-1)\left\{H(\rvv_B | \rvh_{AB}, \rvg_{AE}, \rvz_A, \rvbu) - H(\rvv_B| \rvx_A, \rvbu)\right\}\label{eq:iid-xy}\\
&=(T-1)\left\{I(\rvv_B; \rvx_A, \rvbu) - I(\rvv_B; \rvh_{AB}, \rvg_{AE}, \rvz_A, \rvbu)\right\}\notag\\
&=(T-1)\left\{I(\rvv_B; \rvx_A, \rvbu) - I(\rvv_B; \rvh_{AB}, \rvg_{AE}, \rvz_A)\right\}\label{eq:u-markov2}
\end{align}
where~\eqref{eq:iid-xy} follows from the fact that $\brvbx_A \in {\mathbb C}^{T-1}$ is sampled i.i.d.\ $p_{\rvx_A}(\cdot)$
~\eqref{eq:u-markov1} and~\eqref{eq:u-markov2} follow from the fact that $\rvbu\rightarrow (\brvbz_A, \rvg_{AE}, \rvh_{AB}) \rightarrow \brvby_B \rightarrow \brvbv_B$ holds.

In a similar way we can show that
\begin{align}
&I(\brvbv_A;\brvbx_B,\rvbu)  - I(\brvbv_A; \rvh_{BA},\rvg_{BE}, \brvbz_B) \notag\\
&=(T-1)\left\{I(\rvv_A; \rvx_B, \rvbu) - I(\rvv_A; \rvh_{AB},\rvg_{BE}, \rvz_B)\right\} \label{eq:lb-t3}
\end{align}
By substituting~\eqref{eq:u-markov2} and~\eqref{eq:lb-t3} into~\eqref{eq:lb-t1} we have that
\begin{align}
&\frac{1}{N}H(\rvu_{BA}^K,\rvu_{AB}^K, \brvbv_A^K, \brvbv_B^K | \phi_A, \phi_B, \psi_A, \psi_B, \brvbz^K, \rvbg^K) \ge \notag\\
&\frac{1}{T}\left\{I(\hh_{AB}; \rvu_{BA}) + I(\hh_{BA}; \rvu_{AB})-I(\rvu_{BA};\rvu_{AB})\right\} + \notag\\
&\frac{(T-1)}{T}\left\{I(\rvv_B; \rvx_A, \rvbu) - I(\rvv_B; \rvh_{AB}, \rvg_{AE}, \rvz_A)\right\} + \notag\\
&\frac{(T-1)}{T}\left\{I(\rvv_A; \rvx_B, \rvbu) - I(\rvv_A; \rvh_{AB},\rvg_{BE}, \rvz_B)\right\} \g_K(\eps)\notag\\
&\defeq R_\mrm{eq}- \g_{K}(\eps)
\label{eq:Rnodisc-LB}
\end{align}

\subsection{Secret-Key Generation}
\label{subsec:sec-key-gen}
It remains to show that a secret-key can be generated at the legitimate terminals that achieves a rate $R_\mrm{eq}$ in~\eqref{eq:Rnodisc-LB}
and  satisfies the secrecy constraint.

We consider $M$ i.i.d.\ repetitions of the coding scheme in the previous section. Let \begin{align}\O^M = (\rvbv_A^{M\cdot K}, \rvbv_B^{M\cdot K}, \rvu_{AB}^{M\cdot K}, \rvu_{BA}^{M\cdot K} )\end{align} denote the common sequence at the legitimate terminals after $M$ such repetitions. Note that $\O^M$ is sampled i.i.d.\ from the joint distribution $p(\rvbv_A^K, \rvbv_B^K, \rvu_{AB}^K,\rvu_{BA}^K)$.  Similarly let
\begin{align}
\T^M \defeq \left(\phi_A^M, \phi_B^M, \psi_A^M, \psi_B^M, \brvbz^{M\cdot K}, \rvbg^{M\cdot K}\right)
\end{align}
denote the observation sequence at the eavesdropper which is also sampled i.i.d.\ and furthermore from~\eqref{eq:Rnodisc-LB} we have that
\begin{align}
\frac{1}{K}H(\O | \T) \ge R_\mrm{eq} - \g_{K,\eps}. \label{eq:equiv-lb}
\end{align}

Given the common sequence $\O^M$ at the legitimate receivers and the sequence $\T^M$ at the eavesdropper, sampled i.i.d.\ from a joint distribution, for any rate $R < H(\O | \T)$ and $\delta > 0$ there exists a secret-key codebook  (see e.g.,\cite[sec.~22.3.2, pp.~562-563]{ElGamalKim}) ~such that the  secret-key $\rvk$ satisfies
\begin{align*}
\lim_{M \rightarrow \infty} \frac{1}{M K} H(\rvk) \ge (R - \delta),\quad
\lim_{M\rightarrow \infty}\frac{1}{M K}I(\rvk; \T^M) \le \delta.
\end{align*}

Furthermore for each $M$, by selecting $K$ sufficiently large, we can make $\g_{K,\eps}$  sufficiently small and also guarantee that the error probability in decoding $\O^M$ goes to zero.  Thus we can achieve a secret-key rate arbitrarily close to $R_\mrm{eq}$. This completes the proof of our lower bound.

\subsection{Proof of Prop.~\ref{thm:lb}}

\label{subsec:nodisc-G}
We  evaluate the rate expression in Theorem~\ref{thm:lb-nodisc} for jointly Gaussian input distribution. 
Recall that $\hh_{AB}$ and $\hh_{BA}$ are the MMSE estimates of the channel gains
$\rvh_{AB}$ and $\rvh_{BA}$ respectively. Hence we can decompose,
\begin{align}
\rvh_{AB} = \hh_{AB} + \rve_{AB},\quad
\rvh_{BA} = \hh_{BA} + \rve_{BA}
\end{align}
where $\hh_{AB} \sim \CN(0,\al)$,  $\hh_{BA} \sim \CN(0,\al)$ are the estimates
and  $\rve_{AB} \sim \CN(0,1-\al),$  $\rve_{BA} \sim \CN(0,1-\al)$ are the estimation errors
and where $\al = \frac{P_1}{P_1+1}$. For the test channels in~\eqref{eq:u-test}
we can show the following:
{\allowdisplaybreaks{\begin{align}
&I(\rvu_{AB}; \hh_{AB}) = I(\rvu_{BA}; \hh_{BA}) = \log\left(1 + \frac{\al}{Q_1}\right)\label{eq:uAB_hAB}\\
&I(\rvu_{AB}; \hh_{BA})\!\!= \!\!I(\rvu_{BA};\hh_{AB}) \!\!= -\log\left(1- \frac{\al^2 \rho^2}{1+\frac{Q_1}{\al}}\right)\label{eq:uAB_hBA}\\
&I(\rvu_{AB};\hh_{AB}) - I(\rvu_{AB}; \hh_{BA}) \notag \\ &\qquad=\log\left(1 + \frac{\al(1-\al^2\rho^2)}{Q_1}\right) \notag \\
&\qquad\quad= I(\rvu_{BA}; \hh_{BA})- I(\rvu_{BA}; \hh_{AB}) \label{eq:uh-diff}\\
&I(\rvu_{AB}; \rvu_{BA}) = -\log\left(1 - \frac{\al^2 \rho^2}{\left(1+ \frac{Q_1}{\al}\right)^2}\right)
\end{align}}}

Substituting in~\eqref{eq:RI},~\eqref{eq:Rc-u1} and~\eqref{eq:Rc-u2} we have that
\begin{align}
&R_I \!= \!\!\frac{1}{T}\bigg\{I(\rvu_{AB}; \hh_{BA}) \!\!+ \!\! I(\rvu_{BA}; \hh_{AB}) \!\!- \!\!I(\rvu_{AB}; \rvu_{BA})\bigg\}\notag\\
&= \frac{1}{T}\left\{\!\!-2\log\left(1-\frac{\al^2 \rho^2}{1+\frac{Q_1}{\al}}\right)\!\! +\!\! \log\left(1-\frac{\al^2 \rho^2}{\left(1+\frac{Q_1}{\al}\right)^2}\!\!\right)  \right\}
\label{eq:Ra}
\end{align}
where $Q_1$ is selected to satisfy:
\begin{align}
\log\left(1 + \frac{\al(1-\al^2\rho^2)}{Q_1}\right) \le \eps_1 (T-1)R_\mrm{NC}(P).
\label{eq:Ra-cons}
\end{align}

We next obtain expression for $R_\mrm{BA}^-$ using the test channel~\eqref{eq:v-test}.
\begin{align}
&R_\mrm{BA}^- = I(\rvv_A; \rvx_B,\rvu_{AB},\rvu_{BA})-I(\rvv_A; \rvh_{AB}, \rvz_B, \rvg_{BE})\\
&=h(\rvv_A|\rvh_{AB},\rvz_B, \rvg_{BE}) -  h(\rvv_A|\rvu_{AB},\rvu_{BA}, \rvx_B) \label{eq:Rba-lb1}
\end{align}
we bound each of the two entropy terms as follows. 
\begin{align}
&h(\rvv_A| \rvh_{AB}, \rvz_B, \rvg_{BE})\\
&=E\left[\log\left(1 + Q_2 + \frac{P_2 |\rvh_{BA}|^2}{1+P_2 |\rvg_{BE}|^2}\right)\right] + \log(2\pi e)\label{eq:Rba-lb2}
\end{align}
and furthermore
\begin{align}
&h(\rvv_A|\rvu_{AB},\rvu_{BA}, \rvx_B) \notag \\&\le h(\tilde{\rvh}_{BA} \rvx_B + \rvq_{BA} + \rvn_{BA})\\
&\le \log\left(\sigma_{AB}^2 P_2 + Q_2 + 1\right) + \log (2\pi e)\label{eq:Gauss_UB}.
\end{align} where $$\tilde{\rvh}_{BA} =\rvh_{BA} - \hat{\rvh}_{BA}$$ is the error in the MMSE estimate of $\rvh_{BA}$ given $(\rvu_{AB}, \rvu_{BA})$
and $\sigma_{AB}^2$ is the associated estimation error
\begin{align}
\sigma_{AB}^2 = E[(\rvh_{AB}- \hat{\rvh}_{AB})^2] 
\end{align}The upper bound in~\eqref{eq:Gauss_UB} follows from the fact that a 
Gaussian input maximizes the differential entropy for a given noise variance. 
The expression in~\eqref{eq:R-AB} follows by substituting~\eqref{eq:Rba-lb2}
and~\eqref{eq:Gauss_UB}  into~\eqref{eq:Rba-lb1}.
In a similar way we can establish~\eqref{eq:R-BA}.
The  rate constraint in~\eqref{eq:Rc-v1} and~\eqref{eq:Rc-v2} follow by substituting~\eqref{eq:v-test} into
\begin{align}
I(\rvv_A; \rvy_A| \rvx_A, \rvu_{AB}, \rvu_{BA}) \le \eps_2 R_\mrm{NC}(P).\label{eq:Rb-cons-3}
\end{align}
Note that
\begin{align}
&I(\rvv_A; \rvy_A| \rvx_A, \rvu_{AB}, \rvu_{BA}) \notag \\
&= h(\rvv_A | \rvx_A, \rvu_{AB}, \rvu_{BA})- h(\rvq_{BA})\\
&\le \log\left(1 + \frac{\sigma_{AB}^2 P_2 + 1}{Q_1} \right).\label{eq:Rb-cons-2}
\end{align} The expression~\eqref{eq:Rb-cons}  follows from~\eqref{eq:Rb-cons-3} and~\eqref{eq:Rb-cons-2}.

Finally the bound on $\sigma_{AB}^2$  in~\eqref{eq:sigma} is obtained as follows
\begin{align}
\sigma_{AB}^2 &= E[(\rvh_{AB}- \hh_{AB}(\rvu_{AB}, \rvu_{BA}))^2]\notag \\
&\le E[(\rvh_{AB}-\hh_{AB}(\rvu_{AB}))^2]\le
1- \frac{\al^2}{Q_1 + \al}
\end{align}
This completes the proof of Prop.~\ref{thm:lb}.

\section{Proof of Corollary~\ref{corol:highSNR} }
\label{sec:cap-highSNR}
We establish~\eqref{eq:c1} and~\eqref{eq:c2} respectively. 

We consider the expression for $R^+$ in Theorem~\ref{thm:ub} and show that
\begin{align}
\lim_{P\rightarrow\infty}R^+ = -\frac{1}{T}\log(1-\rho^2) + \g \label{eq:highsnr-ub}
\end{align}
where $\g$ is defined in~\eqref{eq:g-def}. In particular note that the first term in~\eqref{eq:R+} is independent of $P$ and furthermore since $\rvh_{AB}$ and $\rvh_{BA}$ are jointly Gaussian, zero mean, unit variance and with a cross-correlation of $\rho$ it can be readily shown that
\begin{align}
I(\rvh_{AB};\rvh_{BA}) = -\log(1-\rho^2).
\end{align} Only the second and third terms depend on $P$. For any power allocation $P(\rvh_{AB})$ note that
\begin{align}
E\left[\log\left(1 +\frac{P(\rvh_{AB})|\rvh_{AB}|^2}{1+P(\rvh_{AB})|\rvg_{AE}|^2}\right)\right] \le E\left[\log\left(1 + \frac{|\rvh_{AB}|^2}{|\rvg_{AE}|^2}\right)\right]\label{eq:highsnr-ub-1}
\end{align}
which follows since the function $f(x) = \frac{ax}{1+bx}$ is increasing in $x$ for any $a, b >0$. Similarly we have
\begin{align}
E\left[\log\left(1 + \frac{P(\rvh_{BA})|\rvh_{BA}|^2}{1+P(\rvh_{BA})|\rvg_{BE}|^2}\right)\right] \le E\left[\left(1+\frac{|\rvh_{BA}|^2}{|\rvg_{BE}|^2}\right)\right].
\label{eq:highsnr-ub-2}
\end{align} 
Furthermore since a Gaussian input distribution maximizes the conditional mutual information (c.f.~\eqref{eq:gauss-opt}), the upper bound follows by substituting~\eqref{eq:highsnr-ub-1} and~\eqref{eq:highsnr-ub-2} into~\eqref{eq:R+}. 

We next show that,
\begin{align}
\lim_{P\rightarrow\infty}R_\mrm{PD}^- = -\frac{1}{T}\log(1-\rho^2) + \frac{T-1}{T}\g.\label{eq:highsnr-lb-disc}
\end{align}
In particular we consider the secret-key rate expression~\eqref{eq:pdisc} in Prop.~\ref{prop:pdisc}. We select $P_2 = \frac{\sqrt{P}}{T-1}$
and $P_1 = P-\sqrt{P}$. Note that as $P \rightarrow\infty$ we have that $P_1, P_2 \rightarrow \infty$ and $\frac{P_2}{P_1}\rightarrow 0$.
In particular, since $\al = \frac{P_1}{P_1+1}$ we have that
\begin{align}
&\lim_{P_1\rightarrow\infty} -\log(1-\al^2\rho^2) = -\log(1-\rho^2),\label{eq:highsnr-lbdisc-t1}\\
&\lim_{P\rightarrow\infty} \log\left(1+\frac{P_2}{1+P_1}\right)=0. \label{eq:highsnr-penalty-lbdisc}
\end{align}
Furthermore since the function $f(x) = \frac{ax}{1+bx}$ is bounded for all $a, b > 0$ we can apply the Dominated Convergence Theorem,
\begin{align}
\lim_{P_2 \rightarrow\infty}\! E\left[\log\left(1 + \frac{P_2|\rvh_{AB}|^2}{1+ P_2 |\rvg_{AE}|^2}\right)\right] \!\!=\!\! E\left[\log\left(1 + \frac{|\rvh_{AB}|^2}{|\rvg_{AE}|^2}\right)\right] \label{eq:highsnr-c1-lbdisc}\\
\lim_{P_2 \rightarrow\infty} \! E\left[\log\left(1 + \frac{P_2|\rvh_{BA}|^2}{1+ P_2 |\rvg_{BE}|^2}\right)\right] \!\!=\!\! E\left[\log\left(1 + \frac{|\rvh_{BA}|^2}{|\rvg_{BE}|^2}\right)\right] \label{eq:highsnr-c2-lbdisc}
\end{align}
Substituting~\eqref{eq:highsnr-lbdisc-t1},~\eqref{eq:highsnr-penalty-lbdisc},~\eqref{eq:highsnr-c1-lbdisc},~\eqref{eq:highsnr-c2-lbdisc} into
\eqref{eq:pdisc}, \eqref{eq:Rp_T}, \eqref{eq:Rp_AB}, \eqref{eq:Rp_BA} we recover~\eqref{eq:highsnr-lb-disc}.
Comparing~\eqref{eq:highsnr-lb-disc} and~\eqref{eq:highsnr-ub} we establish~\eqref{eq:c1}. 

To establish~\eqref{eq:c2} we consider~\eqref{eq:R-} in Prop.~\ref{thm:lb} and show that with a suitable choice of $Q_1,$ $Q_2,$
$P_1$ and $P_2$ we have that
\begin{align}
\lim_{P\rightarrow\infty} R^- = -\frac{1}{T} \log(1-\rho^2) + \frac{T-1}{T}\g. \label{eq:highsnr-lb-nodisc}
\end{align}
In particular we select $P_2 = \frac{\sqrt{P}}{T-1}$ and $P_1 = P-\sqrt{P}$. Furthermore we select
\begin{align}
Q_1 &= \frac{\al(1-\al^2\rho^2)}{\eps_1 (T-1) R_\mrm{NC}(P)}\\
Q_2 &=\frac{\sigma^2_{AB}P_2+1}{\eps_2 R_\mrm{NC}(P)}
\end{align}
Note that as $P\rightarrow \infty$, it is well known that $R_\mrm{NC}(P)\rightarrow \infty$.  Thus we can let $\eps_1$ and $\eps_2$
to be any functions of P such that $\eps_1(P), \eps_2(P) \rightarrow 0$ and $\eps_i(P) R_\mrm{NC}(P) \rightarrow \infty$.  Then observe that 
$Q_1, Q_2 \rightarrow 0$ as $P \rightarrow \infty$
It therefore follows that\begin{align}
\lim_{P\rightarrow\infty}R_I(\eps_1, P) = -\frac{1}{T}\log(1-\rho^2),\label{eq:highsnr-lb-nodisc0}\\
\lim_{P\rightarrow\infty}\log\left(1+\frac{\sigma_{AB}^2 P_2}{1+Q_2}\right)=0.
\end{align}
and hence we have
\begin{align}
&\lim_{P \rightarrow\infty} R_\mrm{AB}^- =  E\left[\log\left(1 + \frac{|\rvh_{AB}|^2}{|\rvg_{AE}|^2}\right)\right]\label{eq:highsnr-lb-nodisc1}\\
&\lim_{P\rightarrow\infty} R_\mrm{BA}^- = E\left[\log\left(1 + \frac{|\rvh_{BA}|^2}{|\rvg_{BE}|^2}\right)\right]\label{eq:highsnr-lb-nodisc2}
\end{align}
By substituting~\eqref{eq:highsnr-lb-nodisc0},~\eqref{eq:highsnr-lb-nodisc1} and~\eqref{eq:highsnr-lb-nodisc2} into~\eqref{eq:R-} we obtain~\eqref{eq:highsnr-lb-nodisc}.

\section{Conclusion}

We study secret-key agreement capacity over  a two-way, non-coherent, reciprocal, block-fading channel. Our main observation is that a separation based scheme that judiciously combines channel-training and randomness-sharing techniques is an efficient approach when the  coherence block length is large and the high signal-to-noise is high. Numerical results indicate that even when the overhead of transmitting public messages is accounted for, the proposed scheme  outperforms a idealized training-only schemes for moderate SNR and relatively small coherence block-lengths. These observations suggest that when the channel fluctuations are relatively slow, one should not rely on channel reciprocity alone for secret-key generation. We also establish an upper bound on the secret-key agreement capacity and show that the upper and lower bounds coincide asymptotically.

In terms of future work a number of interesting directions remain to be explored. The secret-key agreement capacity in the low signal to noise ratio regime remains to be studied.  Potentially improved upper and lower bounds for moderate SNR may be also obtained. While our upper bound and lower bounds naturally extend to the case of multiple antennas a detailed analysis is left for future work. We note that some preliminary investigation along these lines has appeared in~\cite{matias}.

\appendices

\section{Proof of Lemma~\ref{lem:cond-ent}}
\label{app:cond-ent}
We lower bound the conditional entropy term as follows
\begin{align}
&H(\rvy^n|\rvu^n(L))  = H(\rvy^n, \rvx^n|\rvu^N(L)) -H(\rvx^n|\rvy^n,\rvu^n(L))\\
&\ge H(\rvy^n,\rvx^n)-H(\rvu^n(L))  -H(\rvx^n|\rvy^n,\rvu^n(L))\\
&=nH(\rvy,\rvx)  -H(\rvu^n(L)) - H(\rvx^n|\rvy^n,\rvu^n(L))\label{eq:xy-indep}\\
&\ge nH(\rvy,\rvx)  - nR - H(\rvx^n|\rvy^n,\rvu^n(L))\label{eq:u-bnd}
\end{align}
where~\eqref{eq:xy-indep} follows from the fact that $(\rvx^n,\rvy^n)$ are sampled independently of $\cA$
and~\eqref{eq:u-bnd} follows from the fact that $|\cA| =2^{nR}$. It only remains to upper bound
the conditional entropy term in~\eqref{eq:u-bnd}. Let ${\mathbb I}_{\eps'}$ denote an indicator function
that equals $1$ if $(\rvx^n,\rvy^n,\rvu^n) \in T_{\eps'}^n(\rvx,\rvy,\rvu)$ and equals zero otherwise. From
the conditional typicality lemma we have that $\Pr({\mathbb{I}}_{\eps'}=1) \rightarrow 1$ as $n\rightarrow \infty$,
for any $\eps' > \eps$.

\begin{align}
&H(\rvx^n|\rvy^n,\rvu^N(L)) \le 1 + H(\rvx^n|\rvy^n,\rvu^N(L), {\mathbb I})  \\
&=1 + H(\rvx^n|\rvy^n,\rvu^N(L), {\mathbb I}=1) \Pr({\mathbb I}=1) \notag\\ &\qquad+ H(\rvx^n|\rvy^n,\rvu^N(L), {\mathbb I}=0) \Pr({\mathbb I}=0)\\
&\le 1 + H(\rvx^n|\rvy^n,\rvu^N(L), {\mathbb I}=1) + n H(\rvx)\Pr({\mathbb I}=0)\label{eq:I-zero}\\
&\le 1 + E[\log |\cT_{\eps'}^n(\rvx|\rvu^n, \rvy^n)|]+ n H(\rvx)\Pr({\mathbb I}=0)\label{eq:typ-set}\\ 
&\le 1 + nH(\rvx|\rvu,\rvy) + n\g(\eps) + n\g_n(\eps)\label{eq:cond-lb-2}
\end{align}
where we use the fact that for any $\eps>0$, $\Pr({\mathbb I}=0) \rightarrow 0$ as $n\rightarrow\infty$ in~\eqref{eq:I-zero}, which follows since $R > I(\rvx; \rvu)$.  Combining~\eqref{eq:u-bnd} and~\eqref{eq:cond-lb-2} we get
\begin{align}
&H(\rvy^n|\rvu^n(L)) \ge n\left\{H(\rvy,\rvx)- H(\rvx|\rvy,\rvu) -R -\g_n(\eps)\right\}\\
&= n\left\{H(\rvy|\rvx) + H(\rvx) - H(\rvx|\rvu) +I(\rvx;\rvy|\rvu) - R - \g_n(\eps)\right\}\\
&=n\left\{H(\rvy|\rvu) + I(\rvx;\rvu)-R - \g_n(\eps)\right\},
\end{align}
where the last step uses the fact that $\rvu \rightarrow \rvx \rightarrow \rvy$.

\bibliographystyle{IEEEtran} 
\bibliography{sm}

\end{document}